\documentclass[11pt]{article}
\usepackage{etoolbox}
\usepackage{covington}
\usepackage{graphicx}
\usepackage{latexsym}
\usepackage{amssymb}
\usepackage{amsmath}
\usepackage{amsthm}
\usepackage{thmtools}
\usepackage{caption}
\usepackage{subcaption}
\usepackage{mwe}
\usepackage{derivative}
\usepackage{forloop}
\usepackage[paper=letterpaper,margin=1in]{geometry}
\usepackage[ruled,vlined,linesnumbered]{algorithm2e}
\usepackage{paralist}
\usepackage{xcolor}
\usepackage{nicefrac}
\usepackage{pgfplots}
\usepackage{siunitx}
\usepackage{multirow}
\usepackage{mleftright}\mleftright
\usepackage{numprint}
\usepackage{tikz}
\usepackage{pgfplotstable}
\usepackage[rightcaption]{sidecap}
\usepackage{hyperref}

\usepgfplotslibrary{fillbetween}
\pgfplotsset{compat=1.17}
\SetDiaOffset{-0.35ex}

\definecolor{darkgreen}{rgb}{0,0.5,0}
\hypersetup{
    unicode=false,            
    colorlinks=true,          
    linkcolor=blue!70!black,  
    citecolor=darkgreen,      
    filecolor=magenta,        
    urlcolor=blue!70!black    
}

\newtheorem{theorem}{Theorem}[section]
\newtheorem{lemma}[theorem]{Lemma}

\newtheorem{corollary}[theorem]{Corollary}

\newtheorem{observation}[theorem]{Observation}

\newcommand{\defcal}[1]{\expandafter\newcommand\csname c#1\endcsname{{\mathcal{#1}}}}
\newcommand{\defbb}[1]{\expandafter\newcommand\csname b#1\endcsname{{\mathbb{#1}}}}
\newcommand{\defvec}[1]{\expandafter\newcommand\csname v#1\endcsname{{\mathbf{#1}}}}
\newcommand{\defmat}[1]{\expandafter\newcommand\csname m#1\endcsname{{\mathbf{#1}}}}
\newcounter{calBbCounter}
\forLoop{1}{26}{calBbCounter}{
    \edef\Letter{\Alph{calBbCounter}}
		\edef\letter{\alph{calBbCounter}}
    \expandafter\defcal\Letter
		\expandafter\defbb\Letter
		\expandafter\defvec\letter
		\expandafter\defmat\Letter
}

\newcommand{\eps}{\varepsilon}
\newcommand{\nnR}{{\bR_{\geq 0}}}
\newcommand{\nnRE}[1]{{\bR_{\geq 0}^{#1}}}
\newcommand{\characteristic}{{\mathbf{1}}}
\newcommand{\RSet}{{\mathtt{R}}}
\newcommand{\norm}[1]{\left\lVert#1\right\rVert}
\newcommand{\angel}[1]{\left\langle#1\right\rangle}
\DeclareMathOperator*{\argmax}{\arg\max}
\DeclareMathOperator*{\argmin}{\arg\min}

\newcommand{\vzero}{\bar{0}}

\newcommand{\quadprogIP}{{\textsc{quadprogIP}}}

\newcommand{\NMFW}{{\textup{\texttt{Non-monotone Frank-\allowbreak Wolfe}}}}
\newcommand{\NMMFW}{{\textup{\texttt{Non-monotone Meta-Frank-\allowbreak Wolfe}}}}

\newcommand{\headerref}[1]{{\texorpdfstring{\ref{#1}}{\ref*{#1}}}}
\newcommand{\vsigma}{\boldsymbol {\sigma}}

\author{
	Loay Mualem\thanks{Computer Science Department, University of Haifa. Email: \href{mailto:loaymua@gmail.com}{loaymua@gmail.com}}
\and
	Moran Feldman\thanks{Computer Science Department, University of Haifa. Email: \href{mailto:moranfe@cs.haifa.ac.il}{moranfe@cs.haifa.ac.il}}
}
\title{Resolving the Approximability of Offline and Online Non-monotone DR-Submodular Maximization over General Convex Sets}

\begin{document}

\maketitle

\begin{abstract}
In recent years, maximization of DR-submodular continuous functions became an important research field, with many real-worlds applications in the domains of machine learning, communication systems, operation research and economics. Most of the works in this field study maximization subject to down-closed convex set constraints due to an inapproximability result by Vondr{\'{a}}k~\cite{vondrak2013symmetry}. However, Durr et al.~\cite{durr2021non} showed that one can bypass this inapproximability by proving approximation ratios that are functions of $m$, the minimum $\ell_\infty$-norm of any feasible vector. Given this observation, it is possible to get results for maximizing a DR-submodular function subject to \emph{general} convex set constraints, which has led to multiple works on this problem. The most recent of which is a polynomial time $\tfrac{1}{4}(1 - m)$-approximation offline algorithm due to Du~\cite{du2022lyapunov}. However, only a sub-exponential time $\tfrac{1}{3\sqrt{3}}(1 - m)$-approximation algorithm is known for the corresponding online problem. In this work, we present a polynomial time online algorithm matching the $\tfrac{1}{4}(1 - m)$-approximation of the state-of-the-art offline algorithm. We also present an inapproximability result showing that our online algorithm and Du's~\cite{du2022lyapunov} offline algorithm are both optimal in a strong sense. Finally, we study the empirical performance of our algorithm and the algorithm of Du~\cite{du2022lyapunov} (which was only theoretically studied previously), and show that they consistently outperform previously suggested algorithms on revenue maximization, location summarization and quadratic programming applications.


\medskip

\end{abstract}

\pagenumbering{arabic}

\section{Introduction}
Optimization of continuous DR-submodular functions has gained prominence in recent times. Such optimization is an important traceable subclass of non-convex optimization, and captures problems at the forefront of machine learning and statistics with many real-world applications (see, e.g.,~\cite{bian2019optimal,hassani2017gradient,mitra2021submodular,soma2017nonmonotone}). The majority of the existing works on DR-submodular optimization (and submodular optimization in general) have been focused either on monotone objective functions, or optimization subject to a down-closed convex set constraint.\footnote{A set $\cK \subseteq [0, 1]^n$ is down-closed if, for every two vectors $\vx, \vy \in [0,1]^n$, $\vx \in \cK$ whenever $\vy \in \cK$ and $\vy$ coordinate-wise dominates $\vx$.} However, many real-world problems are naturally captured as optimization of a non-monotone DR-submodular function over a constraint convex set that is not down-closed. For example, consider a streaming service that would like to produce a summary of recommended movies for a user. Often the design of the user interface places strong bounds on the size of the summary displayed to the user, leading to a non-down-closed constraint. Furthermore, the quality of the summary is often captured by a non-monotone objective since putting very similar films in the summary is detrimental to both its value and professional look.

Motivated by the above-mentioned situation, a few recent works started to consider DR-submodular maximization subject to a general (not necessarily down-closed) convex set constraint $\cK$. In general, no constant approximation ratio can be guaranteed for this problem in sub-exponential time due to an hardness result by Vondr{\'{a}}k~\cite{vondrak2013symmetry}. However, Durr et al.~\cite{durr2021non} showed that this inapproximability result can be bypassed when the convex set constraint $\cK$ includes points whose $\ell_\infty$-norm is less than the maximal value of $1$. Specifically, Durr et al.~\cite{durr2021non} showed a sub-exponential time offline algorithm guaranteeing $\tfrac{1}{3\sqrt{3}}(1 - m)$-approximation for this problem, where $m$ is the minimal $\ell_\infty$-norm of any vector in $\cK$. Later, Th\twodias{\'}{\u}{a}ng \& Srivastav~\cite{thang2021online} showed how to obtain a similar result in an online (regret minimization) setting, and an improved sub-exponential offline algorithm obtaining $\tfrac{1}{4}(1-m)$-approximation was suggested by Du et al.~\cite{du2022improved}. Very recently, Du~\cite{du2022lyapunov} provided the first polynomial time algorithm for this setting, obtaining the same offline $\tfrac{1}{4}(1-m)$-approximation as Du et al.~\cite{du2022improved}. Nevertheless, and despite all the progress described above, there are still important open questions left regarding this setting.
\begin{itemize}
    \item What is the best approximation ratio that can be obtained by a polynomial time offline algorithm? In particular, can such an algorithm guarantee a better than $\tfrac{1}{4}(1 - m)$-approximation, and if not, how much slower must be an algorithm that improves over this approximation ratio.
		\item Is there a polynomial time \emph{online} algorithm guaranteeing any constant approximation ratio? Can such an algorithm match the optimal approximation ratio obtainable by an offline algorithm?
\end{itemize}

In this work we answer all the above questions, which essentially settles the problem of maximizing DR-submodular functions over general convex sets in both the offline and online settings. We also study the empirical performance of the theoretically optimal offline and online algorithms, showing that both algorithms consistently outperform previously suggested algorithms. Below we describe our results in more detail.

\paragraph{Online setting.}

As mentioned above, the state-of-the-art online (regret minimization) algorithm of Th\twodias{\'}{\u}{a}ng \& Srivastav~\cite{thang2021online} achieves $\tfrac{1}{3\sqrt{3}}(1 - m)$-approximation, which it does with sub-exponential running time and roughly $O(\sqrt{T})$-regret, where $T$ is the number of time steps.\footnote{By changing parameter values, it is possible to reduce the time complexity of the algorithm of Th\twodias{\'}{\u}{a}ng \& Srivastav~\cite{thang2021online} to be polynomial. However, this comes at the cost of a regret that is nearly-linear in $T$ and an error term in the approximation ratio that diminishes very slowly (linearly in $\log T$).} In this paper, we describe a new online algorithm improving both the approximation ratio and the time complexity. Specifically, our algorithm achieves $\tfrac{1}{4}(1 - m)$-approximation in polynomial time and roughly $O(\sqrt{T})$-regret. The approximation guarantee of our algorithm matches an inapproximability that we prove for the offline setting (see below), and is thus, optimal. We also study the empirical performance of our algorithm, and show that it outperforms the algorithm of~\cite{thang2021online} on two applications of revenue maximization and location summarization.

\paragraph{Offline setting.}

Recall that the state-of-the-art offline algorithm is a recent polynomial time $\tfrac{1}{4}(1-m)$-approximation algorithm due to Du~\cite{du2022lyapunov}. Our first contribution to the offline setting is an inapproximability result showing that this algorithm is optimal in a very strong sense. Specifically, we show that no sub-exponential time algorithm can significantly improve over this approximation ratio, even when $m$ is fixed to any particular value in $[0, 1]$. Furthermore, since Du~\cite{du2022lyapunov} analyzed only the theoretical performance of his algorithm, it is interesting to study the empirical performance of this algorithm, which we do by considering revenue maximization and quadratic programming applications. 

Coding the algorithm of Du~\cite{du2022lyapunov} for the empirical study is somewhat non-trivial because Du~\cite{du2022lyapunov} presented his algorithm as part of a general mathematical framework for designing algorithms for various submodular optimization problems. Therefore, our empirical study is based on an explicit version of this algorithm that we give in this paper, which is not fully identical to the algorithm of~\cite{du2022lyapunov}. Beside being explicit, our version of the algorithm also has the advantage of being more tuned towards practical performance. For completeness, we include a full analysis of our version of the algorithm of Du~\cite{du2022lyapunov}. This full analysis is also used as a warm-up towards the analysis of our own online algorithm.

\subsection{Related work}
Next, we provide a brief summary of the most relevant results on
DR-submodular maximization. Recently, this field has become the work-horse of numerous applications in the fields of statistics and machine learning, which has lead to a dramatic increase in the number of studies related to it. 

\paragraph{Offline DR-submodular optimization.} 
Bian et al.~\cite{bian2017nonmonotone} considered the problem of maximizing monotone DR-functions subject to a down-closed convex set, and showed that a variant of the Frank-Wolfe algorithm (based on the greedy method proposed by~\cite{calinescu2011maximizing} for set functions) guarantees a $(1-\nicefrac{1}{e})$-approximation for this problem, which is optimal~\cite{nemhauser1978best}. Later, Hassani et al.~\cite{hassani2017gradient} showed that the algorithm of~\cite{bian2017nonmonotone} is not robust in stochastic settings (i.e., when only an unbiased estimator of gradients is available), and proved that gradient methods are robust in such setting while still achieving $\nicefrac{1}{2}$-approximation. When the objective DR-submodular function is not necessarily monotone, the problem becomes harder to approximate. Bian et al.~\cite{bian2019optimal} and Niazadeh et al.~\cite{niazadeh2020optimal} independently provided two algorithms with the same approximation guarantee of $\nicefrac{1}{2}$ for maximizing non-monotone DR-submodular functions over a hypercube, which is optimal~\cite{feige2011maximizing} (the algorithm of~\cite{niazadeh2020optimal} applies also to non-DR submodular functions). For general down-closed convex sets, Bian et al.~\cite{bian2018non} provided a $\nicefrac{1}{e}$-approximation algorithm based on the greedy method of~\cite{feldman2011unified} for set functions. Using the concept of monotonicity ratio, Mualem and Feldman~\cite{mualem2022using} were able to smoothly interpolate between the last result and the $(1 - 1/e)$-approximation obtainable for monotone objectives.

\paragraph{Online DR-submodular optimization.}
Chen et al.~\cite{chen2018online} first considered online optimization of monotone DR-submodular functions over general convex sets (for monotone objective functions, there is no difference between optimization subject to down-closed or general convex sets), and provided two algorithms. One guaranteeing $(1-\nicefrac{1}{e})$-approximation using roughly $O(\sqrt{T})$-regret, and another algorithm which is robust to stochastic settings but guarantees only $\nicefrac{1}{2}$-approximation up to the same regret. Later, Chen et al.~\cite{chen2019projection} presented an algorithm that combines $(1-\nicefrac{1}{e})$-approximation with roughly $O(\sqrt{T})$-regret and robustness, and Zhang et al.~\cite{zhang2019online} showed how one can reduce the number of gradient calculations per time step to one, at the cost of increasing the regret to roughly $O(T^{4/5})$. Such a reduction is important for bandit versions of the same problem. Online optimization of DR-submodular functions that are not necessarily monotone was studied by Thang et al.~\cite{thang2021online}, who provided three algorithms for it. One of these algorithms applies to general convex set constraints, and was already discussed above. Another algorithm applies to maximization over the entire hypercube, and achieves $\nicefrac{1}{2}$-approximation with roughly $O(\sqrt{T})$-regret; and the last algorithm applies to online maximization of non-monotone DR-submodular functions over down-closed convex sets, and achieves $\nicefrac{1}{e}$-approximation with roughly $O(T^{2/3})$-regret.

\subsection{Paper organization}
In Section~\ref{sec:Preliminaries}, we provide some definitions and important properties of DR-submodular functions. Section~\ref{sec:Offline} describes our explicit version of the offline algorithm of Du~\cite{du2022lyapunov}, which also serves as warm up for our novel online algorithm described in Section~\ref{sec:Online}. Our inapproximability result, which shows that the above offline and online algorithms are both optimal, is proved in Section~\ref{sec:Hardness}. Finally, in Section~\ref{sec:Experiments}, we study the empirical performance and robustness of our online algorithm and our version of the algorithm of Du~\cite{du2022lyapunov} by comparing them with previously suggested algorithms on multiple machine learning applications.
\section{Preliminaries}\label{sec:Preliminaries}
DR-submodularity (first defined by~\cite{bian2017guarantees}) is an extension of the submodularity notion from set functions to continuous functions. Formally speaking, given a domain $\cX = \prod_{i=1}^n\mathcal{X}_i$, where $\mathcal{X}_i$ is a closed range in $\bR$ for every $i \in [n]$, a function $F \colon \mathcal{X}\rightarrow\mathbb{R}$ is \emph{DR-submodular} if for every two vectors $\va,\vb\in \cX$, positive value $k$ and coordinate $i\in[n]$, the inequality
\[
	F(\va+k\ve_i)-F(\va)\geq F(\vb+k\ve_i)-F(\vb)
\]
holds whenever $\va\leq \vb$ and $\vb+k\ve_i\in \mathcal{X}$ (here and throughout the paper, $\ve_i$ denotes the standard $i$-th basis vector, and comparison between two vectors should be understood to hold coordinate-wise).
Note that if function $F$ is continuously differentiable, then the above definition of DR-submodulrity is equivalent to 
\[
    \nabla F(\vx)\leq\nabla F(\vy)\quad \forall\;\vx,\vy\in \cX, \vx\geq\vy
		\enspace.
\]
Moreover, when $F$ is twice differentiable, it is DR-submodular if and only if its Hessian is non-positive at every vector $\vx\in\mathcal{X}$.

In this work, we study the problem of maximizing a non-negative DR-submodular function $F\colon 2^\cN \to \nnR$ subject to a general convex body $\cK \subseteq \cX$ (usually polytope) constraint. For simplicity, we assume that $\mathcal{X} = [0,1]^n$. Note that this assumption is without loss of generality since there is a natural mapping from $\mathcal{X}$ to $[0,1]^n$. Additionally, as is standard in the field, we assume that $F$ is $\beta$-smooth for some parameter $\beta > 0$. Recall that $F$ is $\beta$-smooth if it is continuously differentiable, and for every two vectors $\vx, \vy \in [0, 1]^n$, the function $F$ obeys
\[
        \|\nabla F(\vx)-\nabla F(\vy)\|_2 \leq \beta\|\vx-\vy\|_2 \enspace.
\]

In the online (regret minimization) version of the above problem, there are $T$ time steps. In every time step $t \in [T]$, the adversary selects a non-negative $\beta$-smooth DR-submodular function $F_t$, and then the algorithm should select a vector $\vy^{(t)} \in \cK$ without knowing $F_t$ (the function $F_t$ is revealed to the algorithm only after $\vy^{(t)}$ is selected). The objective of the algorithm is to maximize $\sum_{i = 1}^T F_t(\vy^{(t)})$, and its success in doing so is measured compared to the best fixed vector $\vx \in \cK$. More formally, we say that the algorithm achieves an approximation ratio of $c \geq 0$ with regret $\cR(T)$ if
\[
	\bE\left[\sum_{t = 1}^T F_t(\vy^{(t)})\right]
	\geq
	c \cdot \max_{\vx \in \cK} \bE\left[\sum_{t = 1}^T F_t(\vx)\right] - \cR(T)
	\enspace.
\]
The nature of the access that the algorithm has to $F_t$ varies between different versions of the above problem. Some previous works assume access to the exact gradient of $F$. However, our algorithm applies also to a stochastic version of the problem in which only access to an unbiased estimator of this gradient is available.

We conclude this section by introducing some additional notation and two known lemmata that are useful in our proofs. Given two vectors $\vx, \vy \in [0, 1]^n$, we denote by $\vx \vee \vy$ and $\vx \wedge \vy$ their coordinate-wise maximum and minimum, respectively. Using this notation, we can now state the first known lemma, which can be traced back to~\cite{hassani2017gradient} (see Inequality~7.5 in the arXiv version~\cite{hassani2017gradientARXIV} of~\cite{hassani2017gradient}), and is also explicitly stated and proved in~\cite{durr2021non}.

\begin{lemma}[Lemma~1 of~\cite{durr2021non}] \label{lem:local_search_bound}
For every two vectors $\vx,\vy\in [0, 1]^n$ and any continuously differentiable DR-submodular function $F\colon [0, 1]^n\to\bR$,
\begin{align*}
    \angel{\nabla F(x),y-x}\geq F(\vx\vee\vy) + F(\vx\wedge\vy) -2F(\vx)\enspace.
\end{align*}
\end{lemma}

The following lemma originates from a lemma proved by~\cite{feldman2011unified} for set functions. Extensions of this lemma to continuous domains have appeared in~\cite{bian2017nonmonotone,chekuri2015multiplicative}, but for completeness, we include a proof of our exact version of the lemma in Appendix~\ref{sec:norm_less_proof}.
\begin{restatable}{lemma}{lemNormLoss} \label{lem:norm_loss}
For every two vectors $\vx,\vy\in [0, 1]^n$ and any continuously differentiable non-negative DR-submodular function $F\colon [0, 1]^n\rightarrow\nnR$,
\begin{align*}
    F(\vx\vee\vy)\geq(1-\norm{\vx}_\infty)F(\vy)\enspace.
\end{align*}
\end{restatable}

\section{Offline Maximization} \label{sec:Offline}

In this section, we present and analyze an explicit variant of the offline algorithm of Du~\cite{du2022lyapunov} for maximizing a non-negative DR-submodular function $F$ over a general convex set $\cK$. Since the algorithm of Du~\cite{du2022lyapunov} is related to Frank-Wolfe, we name our variant {\NMFW}, and its pseudocode appears as Algorithm~\ref{alg:Offline}. {\NMFW} gets a non-negative integer parameter $T$ and a quality control parameter $\varepsilon \in (0, 1)$.
\begin{algorithm}
\DontPrintSemicolon
Let $\vy^{(0)} \leftarrow \argmin_{\vx\in \mathcal{K}}\norm{\vx}_\infty$.\\
\For{$i=1$ \KwTo $T$}
{
	Let $\vs^{(i)}\leftarrow\argmax_{\vx\in \mathcal{K}}\angel{\nabla F(\vy^{(i-1)}),\vx}$\\
	Let $\vy^{(i)}\leftarrow(1-\eps)\cdot\vy^{(i-1)}+\eps\cdot\vs^{(i)}$
}
\Return the vector maximizing $F$ among $\{\vy^{(0)},\ldots,\vy^{(T)}\}$.
\caption{{\NMFW} $(T,\eps)$\label{alg:Offline}}
\end{algorithm}

For completeness, and as a warmup for Section~\ref{sec:Online}, we present a full analysis of {\NMFW}, independent of the analysis presented by Du~\cite{du2022lyapunov}. The conclusions of our analysis are summarized by the following theorem. We note that, for the purpose of this theorem, it would have sufficed for {\NMFW} to return $\vy^{(T)}$ rather than the best solution among $\vy^{(0)},\dotsc,\vy^{(T)}$. However, returning the best of these solutions results in a better empirical performance at almost no additional cost.
\begin{restatable}{theorem}{thmOffline} \label{thm:Offline}
Let $\mathcal{K}\subseteq [0,1]^n$ be a general convex set, and let $F\colon [0, 1]^n \rightarrow \nnR$ be a non-negative $\beta$-smooth DR-submodular function. Then, {\NMFW} (Algorithm~\ref{alg:Offline}) outputs a solution $\vw \in \mathcal{K}$ obeying
\[
	F(\vw) \geq (1-2\eps)^{T - 1}[(1+\eps)^T-1](1-\min_{\vx\in \mathcal{K}}\norm{\vx}_\infty)\cdot F(\vo)-0.5\eps^2 \beta D^2T \enspace,
\]
where $D$ is the diameter of $\mathcal{K}$ and $\vo \in \arg \max_{\vx\in \mathcal{K}}F(\vx)$. In particular, when $T$ is set to be $\lfloor \ln 2 / \eps \rfloor$,
\[
	F(\vw) \geq (\nicefrac{1}{4} - 3\eps)(1-\min_{\vx\in \mathcal{K}}\norm{\vx}_\infty)\cdot F(\vo) -0.5\eps \beta D^2 \enspace.
\] 
\end{restatable}

We begin the proof of Theorem~\ref{thm:Offline} with the following lemma, which bounds the rate in which the infinity norm of the solution maintained by Algorithm~\ref{alg:Offline} can be increase.
\begin{observation}\label{Observation:norm}
For every integer $0 \leq i \leq T$, $1-\|\vy^{(i)}\|_\infty\geq (1-\eps)^{i}\cdot(1-\|\vy^{(0)}\|_\infty)$.
\end{observation}
\begin{proof}
To prove the lemma, we show by induction that for every fixed coordinate $j \in [n]$, we have $1-y^{(i)}_j\geq (1-\eps)^{i}\cdot(1-y^
{(0)}_j)$. For $i = 0$, this inequality trivially holds. Furthermore, assuming this inequality holds for $i - 1$, it also holds for $i$ because
\begin{align*}
1-y_j^i &= 1 - (1-\eps)y_j^{(i-1)}-\eps s_j^{(i)}\\
&\geq 1-(1-\eps)y_j^{(i-1)}-\eps\\
&= (1-\eps)(1-y_j^{(i-1)})\\
&\geq (1-\eps)^i\cdot(1-y^{(0)}_j)
\enspace,
\end{align*}
where the second inequality follows from the induction hypothesis.
\end{proof}

Using the last observation, we can now prove the following lemma about the rate in which the value of $F(y^{(i)})$ increases as a function of $i$.
\begin{lemma}
For every integer $1\leq i\leq T$, $F(\vy^{(i)})\geq (1-2\eps)\cdot F(\vy^{(i-1)}) + \eps(1-\eps)^{i - 1}\cdot(1-\|\vy^{(0)}\|_\infty) \cdot F(\vo) - 0.5\eps^2\beta D^2$. \label{lem:offline}
\end{lemma}
\begin{proof} By the chain rule,
\begin{align*}
F(\vy^{(i)})-F(\vy^{(i-1)})&=F((1-\eps)\cdot\vy^{(i-1)}+\eps\cdot\vs^{(i)})-F(\vy^{(i-1)})\\
&=\int_0^\eps \left. \frac{F((1-z)\cdot\vy^{(i-1)}+z\cdot\vs^{(i)})}{dz}\right|_{z=r}dr\\
&=\int_0^\eps \langle \vs^{(i)}-\vy^{(i-1)},\nabla F((1-r)\cdot\vy^{(i-1)}+r\cdot\vs^{(i)})\rangle dr\\
&\geq\int_0^\eps\left[\langle\vs^{(i)}-\vy^{(i-1)},\nabla F(\vy^{(i-1)})\rangle - r\beta D^2\right]dr\\
&=\eps\cdot\langle \vs^{(i)}-\vy^{(i-1)},\nabla F(\vy^{(i-1)})\rangle-0.5\eps^2\beta D^2,
\end{align*}
where the inequality follows from the $\beta$-smoothness of $F$. Recall now that $s^{(i)}$ is the maximizer found by Algorithm~\ref{alg:Offline} in its $i$-th iteration, and $\vo$ is one of the values in the domain on which the maximum is calculated. Therefore,
\begin{align*}
    F(\vy^{(i)})-F(\vy^{(i-1)})&\geq\eps\cdot\langle \vs^{(i)}-\vy^{(i-1)},\nabla F(\vy^{(i-1)})\rangle-0.5\eps^2\beta D^2\\
    &\geq\eps\cdot\langle \vo-\vy^{(i-1)},\nabla F(\vy^{(i-1)})\rangle-0.5\eps^2\beta D^2\\
    &\geq \eps\cdot\left[F(\vo\vee\vy^{(i-1)})+F(\vo\wedge\vy^{(i-1)})-2F(\vy^{(i-1)})\right]-0.5\eps^2\beta D^2\\
    &\geq \eps\cdot\left[(1-\eps)^{i - 1}\cdot (1 - \|\vy^{(0)}\|_\infty) \cdot F(\vo)-2F(\vy^{(i-1)})\right]-0.5\eps^2\beta D^2.
\end{align*}
where the third inequality follows from Lemma~\ref{lem:local_search_bound}, and the last inequality from Lemma~\ref{lem:norm_loss}, Observation~\ref{Observation:norm} and the non-negativity of $F$. The lemma now follows by rearranging the last inequality.
\end{proof}

We are ready now to prove Theorem~\ref{thm:Offline}.
\begin{proof}[Proof of Theorem~\ref{thm:Offline}]
To see that the second part of the theorem follows from the first part, note that for $T = \lfloor \ln 2 / \eps \rfloor$ and $\eps < 1/4$,
\begin{align*}
	(1 - 2\eps)^{T - 1}[(1 + \eps)^T - 1]
	\geq{} &
	e^{-2\eps T}(1 - 4\eps^2T)[e^{\eps T}(1 - \eps^2T) - 1]\\
	\geq{} &
	e^{-2\ln 2}(1 - 4\eps \ln 2)[e^{\ln 2 - \eps}(1 - \eps \ln 2) - 1]\\
	={} &
	\left(\frac{1}{4} - \eps \ln 2\right)\left[\frac{2 - 2\eps \ln 2}{e^\eps} - 1\right]\\
	\geq{} &
	\left(\frac{1}{4} - \eps\right)\left[\frac{2 - 2\eps}{1 + 2\eps} - 1\right]\\
	={} &
	\left(\frac{1}{4} - \eps\right) \cdot \frac{1 - 4\eps}{1 + 2\eps}\\
	\geq{} &
	\frac{1}{4} - 3\eps
	\enspace.
\end{align*}
For $\eps \geq 1/4$, the second part of the theorem is an immediate consequence of the non-negaitivity of $F$.

It remains to prove the first part of the theorem. We do that by proving by induction the stronger claim that for every integer $0 \leq i \leq T$,
\begin{equation} \label{eq:offline_induction}
    F(\vy^{(i)})
		\geq
		(1-2\eps)^{i - 1}\left[(1+\eps)^i-1\right]\cdot(1-\|\vy^{(0)}\|_\infty)\cdot F(\vo) - 0.5\eps^2\beta D^2 i
		\enspace.
\end{equation}
Note that the theorem indeed follows from this claim because $w$ is the best vector within a set that includes $\vy^{(T)}$, and $\vy^{(0)} \in \arg\min_{\vx \in \cK} \|x\|_\infty$. For $i = 0$, Equation~\eqref{eq:offline_induction} follows directly from the non-negativity of $F$. Hence, we only need to show that for $1\leq i\leq T$, if we assume that Equation~\eqref{eq:offline_induction} holds for $i - 1$, then it holds for $i$ as well. This is indeed the case because Lemma~\ref{lem:offline} yields
\begin{align*}
    F(\vy^{(i)})
		\geq{} \mspace{-10mu}&\mspace{10mu}
		(1-2\eps)\cdot F(\vy^{(i-1)})+\eps(1-\eps)^{i - 1}\cdot(1-\|\vy^{(0)}\|_\infty)\cdot F(\vo) - 0.5\eps^2\beta D^2\\
		\geq{} &
		(1-2\eps)\cdot \{(1-2\eps)^{i - 2}\left[(1+\eps)^{i - 1}-1\right]\cdot(1-\|\vy^{(0)}\|_\infty)\cdot F(\vo) - 0.5\eps^2\beta D^2 (i - 1)\} \\&+\eps(1-\eps)^{i - 1}\cdot(1-\|\vy^{(0)}\|_\infty)\cdot F(\vo) - 0.5\eps^2\beta D^2\\
    \geq{} &
		\{(1-2\eps)^{i - 1}\left[(1+\eps)^i - \eps(1+\eps)^{i - 1}-1\right] + \eps(1-\eps)^{i - 1}\}\cdot(1-\|\vy^{(0)}\|_\infty)\cdot F(\vo) - 0.5\eps^2\beta D^2 i\\
		\geq{} &
		(1-2\eps)^{i - 1}\left[(1+\eps)^i-1\right]\cdot(1-\|\vy^{(0)}\|_\infty)\cdot F(\vo) - 0.5\eps^2\beta D^2 i
		\enspace,
\end{align*}
where the second inequality follows from the induction hypothesis, and the last inequality holds since
\[
	(1-2\eps)^{i - 1}\cdot \eps(1+\eps)^{i - 1}
	=
	\eps(1 - \eps - 2\eps^2)^{i - 1}
	\leq
	\eps(1 - \eps)^{i - 1}
	\enspace.
	\qedhere
\]
\end{proof}

\section{Online Maximization}\label{sec:Online}

In this section, we consider the problem of maximizing a non-negative DR-submodular function $F$ over a general convex set $\cK$ in the online setting. The only currently known algorithm for this problem is an algorithm due to~\cite{thang2021online} which guarantees $\frac{1-\min_{\vx\in\mathcal{K}}\norm{\vx}_\infty}{3\sqrt{3}}$-approximation. One drawback of this algorithm is that its regret is roughly $T$ over the logarithm of the running time, and therefore, to make this regret less than nearly-linear in $T$ one has to allow for a super-polynomial time complexity (furthermore, a sub-exponential time complexity is necessary to get a regret of $T^c$ for any constant $c \in (0, 1)$). Our algorithm, given as Algorithm~\ref{alg:Online}, combines ideas from our offline algorithm and the Meta-Frank-Wolfe algorithm suggested in~\cite{chen2018online}, and guarantees both $\tfrac{1}{4}(1 - \min_{\vx\in\mathcal{K}}\norm{\vx}_\infty)$-approximation and roughly $O(\sqrt{T})$-regret in polynomial time.

Like the original Meta-Frank-Wolfe algorithm of~\cite{chen2018online}, our algorithm uses in a black-box manner multiple instances $\cE$ of an online algorithm for linear optimization. More formally, we assume that every instance $\cE$ has the following behavior and guarantee. There are $T$ time steps. In every time step $t \in [T]$, $\cE$ selects a vector $\vu^{(t)} \in \cK$, and then an adversary reveals to $\cE$ a vector $\vd^{(t)}$ that was chosen independently of $\vu^{(t)}$. The algorithm $\cE$ guarantees that
\[
	\bE\left[\sum_{t = 1}^T \langle \vu^{(t)},\vd^{(t)} \rangle\right]
	\geq
	\max_{\vx \in \cK} \bE\left[\sum_{t = 1}^T \langle \vx,\vd^{(t)} \rangle\right] - \cR(T)
\]
for some regret function $\cR(T)$ that depends on the particular linear optimization algorithm chosen as the black-box (and may depend on the convex body $\cK$ and the bounds available on the adversarially chosen vectors $\vd^{(t)}$). One possible choice for an online linear optimization algorithm is Regularized-Follow-the-Leader due to~\cite{abernethy2008competing} that has $\cR(T) \leq DG\sqrt{2T}$, where $D$ is the diameter of $\cK$ and $G = \max_{1 \leq t \leq T} \|\vd^{(t)}\|_2$.

Algorithm~\ref{alg:Online} runs in each time step a procedure similar to our version of the offline algorithm (\NMFW). However, instead of calculating a point $\vs$ that is good with respect to the gradient at the current solution, Algorithm~\ref{alg:Online} asks an instance of an online linear optimization algorithm to provide such a point. At the end of the time step, the online linear optimization algorithm gets an estimate of the gradient as the adversarial vector, and therefore, on average, the points it produces are a good approximation of the optimal point in retrospect. Algorithm~\ref{alg:Online} gets three parameters. The parameters $L$ and $\eps$ correspond to the parameters $T$ and $\eps$ of {\NMFW} (Algorithm~\ref{alg:Offline}),\footnote{The parameter $T$ of {\NMFW} was renamed to $L$ here to accommodate the standard notation in both offline and online algorithms. In offline Frank-Wolfe-like algorithms, the number of iterations is usually denoted by $T$, and in online algorithms $T$ is reserved to the number of time steps.} respectively, and the parameter $T$ is the number of time steps.

\begin{algorithm}
\DontPrintSemicolon
\lFor{$i = 1$ \KwTo $L$}{Initialize an instance $\cE_i$ of some online algorithm for linear optimization.}
\For{$t=1$ \KwTo $T$}
{
	Let $\vy^{(0, t)} \leftarrow \argmin_{\vx\in \mathcal{K}}\norm{\vx}_\infty$.\\
	\For{$i = 1$ \KwTo $L$}
	{
		Let $\vs^{(i, t)} \in \cK \leftarrow$ be the vector picked by $\mathcal{E}_\ell$ in time step $t$.\\
		Let $\vy^{(i, t)}\leftarrow(1-\eps) \cdot \vy^{(i-1, t)} + \eps \cdot \vs^{(i, t)}$.\\
	}
	Play $\vy^{(t)} = \vy^{(L, t)}$.\\
	\For{$i = 1$ \KwTo $L$}
	{
		Observe an unbiased estimator $\vg^{(i, t)}$ of $\nabla F_t(\vy^{(i - 1, t)})$.\\
		Pass $\vg^{(i, t)}$ as the adverserially chosen vector $\vd^{(t)}$ for $\cE_i$.
	}
}
\caption{{\NMMFW} $(L, \eps, T)$\label{alg:Online}}
\end{algorithm}

The main result that we prove regarding the online setting is given by the next theorem.

\begin{restatable}{theorem}{thmOnline} \label{thm:Online}
Let $\mathcal{K}$ be a general convex set with diameter $D$. Assume that for every $1\leq t \leq T$, $F_t\colon [0, 1]^n \to \nnR$ is a $\beta$-smooth DR-submodular function,
then
\[
	\sum_{t=1}^T\mathbb{E}[F_t(\vy^{(t)})]
	\geq
	(1-2\eps)^{T - 1}[(1+\eps)^T-1](1-\min_{\vx\in \mathcal{K}}\norm{\vx}_\infty)\cdot \bE\left[\sum_{t = 1}^T F(\vo)\right] - \eps L \cdot \cR(T) - 0.5\eps^2 \beta D^2TL \enspace,
\]
where $D$ is the diameter of $\mathcal{K}$, $\vo$ is a vector in $\cK$ maximizing $\bE[\sum_{t=1}^T F_t(\vo)]$, and $\cR(T)$ is the regret of the online linear optimization algorithm over the domain $\cK$ when the adversarial vectors $\vd^{(t)}$ are the estimators $\vg^{(i, t)}$ calculated by Algorithm~\ref{alg:Online}. In particular, when $L$ is set to be $\lfloor \ln 2 / \eps \rfloor$, $\eps$ is set to be $1 / \sqrt{T}$ and $\cE_i$ is chosen as an instance of Regularized-Follow-the-Leader,
\[
	\sum_{t=1}^T\mathbb{E}[F_t(\vy^{(t)})]
	\geq
	(\nicefrac{1}{4} - 3\eps)(1-\min_{\vx\in \mathcal{K}}\norm{\vx}_\infty)\cdot \bE\left[\sum_{t=1}^T F_t(\vo)\right] - (G + \beta D) D\sqrt{T} \enspace,
\]
where $G = \max_{1 \leq i \leq L, 1 \leq t \leq T} \|\vg^{(i, t)}\|_2$.
\end{restatable}

\noindent \textbf{Remark:} In the last theorem we have set $\eps$ to a value of $1 / \sqrt{T}$, which requires pre-knowledge of $T$. This can be avoided by using a dynamic value for $\eps$ that changes as a function of the number of time slots that have already passed.

We begin the proof of Theorem~\ref{thm:Online} by observing that a repetition of the first half of the proof of Lemma~\ref{lem:offline} leads to the following lemma.
\begin{lemma} \label{lem:online}
For every two integers $1 \leq t \leq T$ and $1\leq i\leq L$, $F_t(\vy^{(i,t)})\geq F_t(\vy^{(i - 1,t)}) + \eps\cdot\langle \vs^{(i,t)}-\vy^{(i-1,t)},\nabla F_t(\vy^{(i-1,t)}\rangle-0.5\eps^2\beta D^2$.
\end{lemma}

Using the guarantee of $\cE_i$, it is possible to get the following lemma from the previous one.

\begin{lemma} \label{lem:sum}
For every integer $1\leq i\leq L$, $\bE[\sum_{t = 1}^T F_t(\vy^{(i,t)})]\geq \bE[\sum_{t = 1}^T F_t(\vy^{(i - 1,t)}) + \eps \cdot \sum_{t = 1}^T\langle \vo -\vy^{(i-1,t)},\nabla F_t(\vy^{(i-1,t)})\rangle] - \eps \cdot \cR(T) - 0.5\eps^2\beta D^2T$.
\end{lemma}
\begin{proof}
Summing up Lemma~\ref{lem:online} over all $t$ values, we get
\begin{align*}
	\sum_{t = 1}^T F_t(\vy^{(i,t)})
	\geq{} &
	\sum_{t = 1}^T F_t(\vy^{(i - 1,t)}) + \eps\cdot \sum_{t = 1}^T\langle \vs^{(i,t)}-\vy^{(i-1,t)},\nabla F_t(\vy^{(i-1,t)})\rangle-0.5\eps^2\beta D^2T\\
	={} &
	\sum_{t = 1}^T F_t(\vy^{(i - 1,t)}) + \eps\cdot \left[\sum_{t = 1}^T\langle \vs^{(i,t)},\vg^{(i, t)}\rangle + \sum_{t = 1}^T\langle \vs^{(i,t)},\nabla F_t(\vy^{(i-1,t)}) -\vg^{(i, t)}\rangle \right.\\&\mspace{100mu}\left.- \sum_{t = 1}^T\langle \vy^{(i-1,t)},\nabla F_t(\vy^{(i-1,t)})\rangle\right] -0.5\eps^2\beta D^2T
	\enspace.
\end{align*}
Additionally, since $\vg^{(i, t)}$ is independent of $\vs^{(i,t)}$, by the guarantee of $\cE_i$,
\[
	\bE\left[\sum_{t = 1}^T\langle \vs^{(i,t)},\vg^{(i, t)}\rangle\right]
	\geq
	\bE\left[\sum_{t = 1}^T\langle \vo,\vg^{(i, t)}\rangle\right] - \cR(T)
	\enspace.
\]
Finally, since $\vg^{(i, t)}$ is chosen after $\vy^{(i-1,t)}$,
\begin{align*}
	\bE[\langle \vs^{(i,t)},\nabla F_t(\vy^{(i-1,t)}) -\vg^{(i, t)}\rangle \mid \vs^{(i, t)}, \vy^{(i-1,t)}]
	={} &
	\langle \vs^{(i,t)}, \nabla F_t(\vy^{(i-1,t)}) - \bE[\vg^{(i, t)} \mid \vy^{(i-1,t)}]\rangle\\
	={} &
	\langle \vs^{(i,t)}, \nabla F_t(\vy^{(i-1,t)}) - \nabla F_t(\vy^{(i-1,t)})\rangle
	=
	0
	\enspace,
\end{align*}
which by the law of total expectation implies $\bE[\langle \vs^{(i,t)},\nabla F_t(\vy^{(i-1,t)}) -\vg^{(i, t)}\rangle] = 0$.

Combining all the above inequalities yields
\begin{align*}
	\bE\left[\sum_{t = 1}^T F_t(\vy^{(i,t)})\right]
	\geq{} \mspace{-120mu}&\mspace{120mu}
	\bE\left[\sum_{t = 1}^T F_t(\vy^{(i - 1,t)})\right] + \eps\cdot \left\{\sum_{t = 1}^T\langle o, \bE[\vg^{(i, t)}]\rangle - \cR(T) \right.\\&\mspace{220mu}\left.- \bE\left[\sum_{t = 1}^T\langle \vy^{(i-1,t)},\nabla F_t(\vy^{(i-1,t)})\rangle\right]\right\} -0.5\eps^2\beta D^2T\\
	={} &
	\bE\left[\sum_{t = 1}^T F_t(\vy^{(i - 1,t)}) + \eps\cdot \sum_{t = 1}^T\langle o - \vy^{(i-1,t)}, \nabla F_t(\vy^{(i-1,t)})\rangle\right] - \eps \cdot \cR(T) -0.5\eps^2\beta D^2T
	\enspace.
	\qedhere
\end{align*}
\end{proof}

\begin{corollary} \label{cor:online}
For every integer $1\leq i\leq L$, $\bE[\sum_{t = 1}^T F_t(\vy^{(i,t)})]\geq \bE[(1 - 2\eps) \cdot \sum_{t = 1}^T F_t(\vy^{(i - 1,t)}) + \eps(1-\eps)^{i - 1} \cdot \sum_{t = 1}^T (1 - \|\vy^{(0, t)}\|_\infty) \cdot F_t(\vo)] - \eps \cdot \cR(T) - 0.5\eps^2\beta D^2T$.
\end{corollary}
\begin{proof}
To see why this corollary follows from Lemma~\ref{lem:sum}, it suffices to observe that, for every integer $1 \leq t \leq T$,
\begin{align*}
	\langle \vo -\vy^{(i-1,t)},\nabla F_t(\vy^{(i-1,t)})\rangle
	\geq{} &
	F_t(\vo\vee\vy^{(i-1,t)})+F_t(\vo\wedge\vy^{(i-1,t)})-2F_t(\vy^{(i-1,t)})\\
	\geq{} &
	F_t(\vo\vee\vy^{(i-1,t)})-2F_t(\vy^{(i-1)})\\
	\geq{} &
	(1-\eps)^{i - 1}\cdot (1 - \|\vy^{(0,t)}\|_\infty) \cdot F_t(\vo)-2F_t(\vy^{(i-1,t)})
	\enspace,
\end{align*}
where the first inequality follows from Lemma~\ref{lem:local_search_bound}, the second inequality follows from the non-negativity of $F_t$, and the last inequality follows from Lemma~\ref{lem:norm_loss} and the observation that the proof of Observation~\ref{Observation:norm} extends to Algorithm~\ref{alg:Online} and shows that $1 - \|y^{(i, t)}\|_\infty \leq (1-\eps)^{i}\cdot(1-\|\vy^{(0, t)}\|_\infty)$.
\end{proof}

One can observe that Corollary~\ref{cor:online} is very similar to Lemma~\ref{lem:offline} (the main difference between the two is that in Corollary~\ref{cor:online} the sum $\sum_{t = 1}^T F_t$ replaces the function $F$ from Lemma~\ref{lem:offline}). This similarity means that the proof of Theorem~\ref{thm:Offline} can work with Corollary~\ref{cor:online} instead of Lemma~\ref{lem:offline}, which yields Theorem~\ref{thm:Online}.

\section{Inapproximability} \label{sec:Hardness}

In this section, we prove our inapproximability result, which is given by the following theorem. Our result shows that the known offline result (reproved in Section~\ref{sec:Offline}) for maximizing a DR-submodular function subject to a general convex set is optimal. Notice that this implies that our online algorithm from Section~\ref{sec:Online} is also optimal (at least in terms of the approximation ratio) unless one allows for an exponential time complexity.

\begin{theorem} \label{thm:inapproximability}
For every two constants $h \in [0, 1)$ and $\eps > 0 $, no sub-exponential time algorithm can obtain $(\nicefrac{1}{4}(1 - h) + \eps)$-approximation for the problem of maximizing a continuously differentiable non-negative DR-submodular function $F\colon [0, 1]^n \to \nnR$ subject to a solvable polytope $\cK$ obeying $\min_{\vx \in \cK} \|\vx\|_\infty = h$. Furthermore, this is true even if we are guaranteed that $\max_{\vx \in \cK} F(\vx) = \Omega(n^{-1})$ and $F$ is $\beta$-smooth for some $\beta$ that is polynomial in $n$.
\end{theorem}

The last part of Theorem~\ref{thm:inapproximability} specifies some additional conditions under which the inapproximability stated in the theorem still applies. These conditions are important because under them our algorithm from Section~\ref{sec:Offline} can be made to have a clean approximation guarantee of $1/4(1 - \min_{\vx \in \cK} \|\vx\|_\infty) - \eps'$, for any constant $\eps' > 0$, by choosing a polynomially small value for the parameter $\eps$ of the algorithm (to see that this is indeed the case, it is important to observe that since $\cK \subseteq [0, 1]^n$, the diameter $D$ of $\cK$ is at most $\sqrt{n}$).

The proof of Theorem~\ref{thm:inapproximability} is based on the symmetry gap framework of Vondr\'{a}k~\cite{vondrak2013symmetry}. To use this framework, we first need to choose a submodular set function $f_k$ ($k \geq 1$ is an integer parameter of the function). We choose the same function that was used by Vondr\'{a}k~\cite{vondrak2013symmetry} to prove his hardness for maximizing a submodular function subject to a matroid base constraint. Specifically, the ground set of $f_k$ is the set $\cN_k = \{a_i, b_i \mid i \in [k]\}$, and for every set $S \subseteq \cN_k$,
\[
	f_k(S)
	=
	\sum_{i = 1}^k \characteristic[a_i \in S] \cdot \characteristic[b_i \not \in S]
	\enspace.
\]
One can verify that $f_k$ is non-negative and submodular since it is the cut function of a directed graph consisting of $k$ vertex-disjoint arcs.

We now would like to convert $f_k$ into two DR-submodular functions, which we do using the following lemma of~\cite{vondrak2013symmetry}. This lemma refers to the multilinear extension of a set function $f\colon 2^\cN \to \bR$ over a ground set $\cN$. This extension is a function $F\colon [0, 1]^\cN \to \bR$ defined for every vector $\vx \in [0, 1]^\cN$ by $F(\vx) = \bE[f(\RSet(\vx))]$, where $\RSet(\vx)$ is a random subset of $\cN$ that includes every element $u \in \cN$ with probability $x_u$, independently.
\begin{lemma}[Lemma~3.2 of~\cite{vondrak2013symmetry}] \label{lem:original_continuous_versions}
Consider a function $f\colon 2^\cN \to \nnR$ invariant under a group of permutations $\cG$ on the ground set $\cN$. Let $F(\vx)$ be the multilinear extension of $f$, define $\bar{x} = \bE_{\sigma \in \cG}[\characteristic_{\sigma(\vx)}]$ and fix any $\eps' > 0$. Then, there is $\delta > 0$ and functions $\hat{F}, \hat{G} \colon [0, 1]^{\cN} \to \nnR$ (which are also symmetric with respect to $\cG$), satisfying the following:
\begin{compactenum}
	\item For all $\vx \in [0, 1]^\cN$, $\hat{G}(\vx) = \hat{F}(\bar{\vx})$.
	\item For all $\vx \in [0, 1]^\cN$, $|\hat{F}(\vx) - F(\vx)| \leq \eps'$.
	\item Whenever $\|\vx - \bar{\vx}\|_2 \leq \delta$, $\hat{F}(\vx) = \hat{G}(\vx)$ and the value depends only on $\bar{\vx}$.
	\item The first partial derivatives of $\hat{F}$ and $\hat{G}$ are absolutely continuous.\label{part:continuous}
	\item If $f$ is monotone, then, for every element $u \in \cN$, $\frac{\partial\hat{F}}{\partial x_u} \geq 0$ and $\frac{\partial\hat{G}}{\partial x_u} \geq 0$ everywhere.
	\item If $f$ is submodular then, for every two elements $u,v \in \cN$, $\frac{\partial^2\hat{F}}{\partial x_u \partial x_v} \leq 0$ and $\frac{\partial^2\hat{G}}{\partial x_u \partial x_v} \leq 0$ almost everywhere.\label{part:submodular}
\end{compactenum}
\end{lemma}

Observe that $f_k$ is invariant to exchanging the identities of $a_i$ and $b_i$ with $a_j$ and $b_j$, respectively, for any choice of $i, j \in [k]$. Therefore, we can choose $\cG$ in the last lemma as the group of permutations that can be obtained by any number of such exchanges. In the rest of this section, we assume that $\hat{F}_k$ and $\hat{G}_k$ are functions $\hat{F}$ and $\hat{G}$ obtained using Lemma~\ref{lem:original_continuous_versions} for this choice of $\cG$, $f_k$ and $\eps' = 1 / (2k)$. It is also important to note that for this choice of $\cG$ we have for every vector $\vx \in [0, 1]^{\cN_k}$ and $i \in [k]$
\[
	\bar{x}_{a_i} = \frac{1}{k} \sum_{j = 1}^k x_{a_j}
	\qquad
	\text{and}
	\qquad
	\bar{x}_{b_i} = \frac{1}{k} \sum_{j = 1}^k x_{b_j}
	\enspace.
\]

Let us now define a family of polytopes. The polytope $\cP_{h, k}$ is the convex hull of the $k + 1$ vectors $\vv^{(1)}, \vv^{(2)}, \dotsc, \vv^{(k)}$ and $\vu$ defined as follows. For every $j \in [k]$, $u_{a_j} = 0$ and $u_{b_j} = h$. For every $i, j \in [k]$,
\[
	v^{(i)}_{a_j}
	=
	\begin{cases}
		1 & \text{if $i = j$} \enspace,\\
		0 & \text{otherwise} \enspace,
	\end{cases}
	\qquad
	\text{and}
	\qquad
	v^{(i)}_{b_j}
	=
	\begin{cases}
		1 & \text{if $i \neq j$} \enspace,\\
		0 & \text{otherwise} \enspace.
	\end{cases}
\]

Using the above definitions, we can now state two instances of the problem we consider
\[
	\begin{array}{l}
		\max \hat{F}_k(\vx)\\
		\vx \in \cP_{h, k}
	\end{array}
	\qquad
	\text{and}
	\qquad
	\begin{array}{l}
		\max \hat{G}_k(\vx)\\
		\vx \in \cP_{h, k}
	\end{array}
	\enspace.
\]
Below, we refer to these instances as the \emph{basic} instances. We will see that by ``scrambling'' these instances in an appropriate way, they can be made indistinguishable. However, for that to yield Theorem~\ref{thm:inapproximability}, it is necessary to prove that the scrambled instances obey the properties assumed in the theorem, and furthermore, that there is a large gap between the optimal values of scrambled instances derived from the two basic instances. Towards this goal, we first study the properties of the basic instances themselves, and the gap between their optimal values. Let us begin with the following lemma, which gives some properties of the objective functions of the basic instances.
\begin{lemma} \label{lem:hat_functions_properties}
The functions $\hat{F}_k$ and $\hat{G}_k$ are continuously differentiable, non-negative and DR-submodular. Furthermore, they are $\beta$-smooth for a value $\beta$ that is polynomial in $k$.
\end{lemma}
\newtoggle{proofsAppendix}
\togglefalse{proofsAppendix}
\begin{proof}
The non-negativity of $\hat{F}_k$ and $\hat{G}_k$ is explicitly guaranteed by Lemma~\ref{lem:original_continuous_versions}, and Part~\ref{part:continuous} of the lemma shows that $\hat{F}$ and $\hat{G}$ are also continuously differentiable. Finally, Parts~\ref{part:continuous} and~\ref{part:submodular} of Lemma~\ref{lem:original_continuous_versions} imply together that $\hat{F}_k$ and $\hat{G}_k$ are DR-submodular (see the proof of Lemma~3.1 of~\cite{vondrak2013symmetry} for a formal argument).

It remains to bound the smoothness of $\hat{F}_k$ and $\hat{G}_k$. Notice that the following claim implies that both functions are $\beta$-smooth for a $\beta$ value that is polynomial in $k$. Unfortunately, the proof of this claim is technically quite involved (and not very insightful) as it requires us to look into the proof Lemma~\ref{lem:original_continuous_versions}, and therefore, we defer the proof of this claim to Appendix~\ref{app:missing_inapproximability}.
\begin{restatable}{claim}{clmPartialDerivativesBound} \label{clm:partial_derivatives_bound}
The absolute values of the second order partial derivatives of the functions $\hat{F}_k$ and $\hat{G}_k$ are bounded by $16k + 2$ almost everywhere, and therefore, both functions are $\beta$-smooth for a $\beta$ value that is polynomial in $k$. \iftoggle{proofsAppendix}{}{\qedhere}
\end{restatable}
\end{proof}

Next, we observe that the common constraint polytope of the basic instances is solvable since $\cP_{h, k}$ is a polytope over $2k$ variables defined as the convex-hall of $k + 1$ vectors. The next observation proves another property of this polytope.

\begin{observation} \label{obs:min_h}
If $k \geq 1 / (1 - h)$, $\min_{\vx \in \cP_{h, k}} \|\vx\|_\infty = h$.
\end{observation}
\begin{proof}
Since $\vu \in \cP_{h, k}$, $\min_{\vx \in \cP_{h, k}} \|\vx\|_\infty \leq h$. Thus, we only need to show that no point in $\cP_{h, k}$ has an infinity norm less than $h$. Recall that every point in $\cP_{h, k}$ is a convex combination $\sum_{i = 1}^k c_i \vv^{(i)} + d \vu$ (where $c_i$ is the coefficient of $\vv^{(i)}$ in the combination, and $d$ is the coefficient of $\vu$), and assume without loss of generality that $c_1 = \min \{c_1, c_2, \dotsc, c_k\}$. Then,
\[
	\left\|\sum_{i = 1}^k c_i \vv^{(i)} + d \vu\right\|_\infty
	\geq
	\sum_{i = 1}^k c_i v^{(i)}_{b_1} + d u_{b_1}
	=
	\sum_{i = 2}^k c_i + dh
	\geq
	\frac{k - 1}{k} \sum_{i = 1}^k c_i + dh
	\geq
	h \sum_{i = 1}^k c_i + dh
	=
	h
	\enspace,
\]
where the last inequality holds by the condition of the observation, and the last equality holds since the fact that $\sum_{i = 1}^k c_i \vv^{(i)} + d \vu$ is a convex combination implies $\sum_{i = 1}^k c_i + d = 1$.
\end{proof}

The last properties that we need to prove for the basic instances are about the optimal values of these instances. Specifically, we need to show that both their optimal values are significant (at least $\Omega(k^{-1})$), but there is a large gap between them. The following two lemmata show these properties, respectively.
\begin{lemma}
$\max_{\vx \in \cP_{h, k}} \hat{F}_k(\vx) = \Omega(k^{-1})$ and $\max_{\vx \in \cP_{h, k}} \hat{G}_k(\vx) = \Omega(k^{-1})$.
\end{lemma}
\begin{proof}
We prove the lemma by considering the vector $\vy = \frac{1}{k} \sum_{i = 1}^k \vu^{(i)}$. Since $\vy \in \cP_{h, k}$ and $\bar{\vy} = \vy$, $\hat{F}_k(\vy)$ lower bounds both $\max_{\vx \in \cP_{h, k}} \hat{F}_k(\vx)$ and $\max_{\vx \in \cP_{h, k}} \hat{G}_k(\vx)$. Thus, it remains to show that $\hat{F}_k(\vy) = \Omega(k^{-1})$. By Lemma~\ref{lem:original_continuous_versions},
\[
	\hat{F}_k(\vy)
	\geq
	F_k(\vy) - \eps'
	=
	\sum_{i = 1}^k y_{a_i}(1 - b_i) - \eps'
	=
	\sum_{i = 1}^k \frac{1}{k} \cdot \left(1 - \left(1 - \frac{1}{k}\right)\right) - \eps'
	=
	\frac{1}{k} - \eps'
	=
	\frac{1}{2k}
	\enspace,
\]
where $F_k$ is the multilinear extension of $f_k$.
\end{proof}

\begin{lemma} \label{lem:opt_bounds}
$\max_{\vx \in \cP_{h, k}} \hat{F}_k(\vx) \geq 1 - 1/(2k)$ and $\max_{\vx \in \cP_{h, k}} \hat{G}_k(\vx) \leq (1 - h)/4 + 3/(2k)$.
\end{lemma}
\begin{proof}
To prove the first part of the lemma, it suffices to observe that $\vv^{(1)} \in \cP_{h, k}$ and
\[
	\hat{F}_k(\vv^{(1)})
	\geq
	F_k(\vv^{(1)}) - \eps'
	=
	f_k(\{a_1\} \cup \{b_i \mid 2 \leq i \leq k\}) - \eps'
	=
	1 - 1/(2k)
	\enspace,
\]
where $F_k$ is the multilinear extension of $f_k$.

Let us now prove the second part of the lemma. Fix an arbitrary vector $\vx \in \cP_{h, k}$, and let $d$ be the coefficient of $\vu$ in the convex combination that shows that $\vx$ belongs to $\cP_{h, k}$. Then,
\[
	\sum_{i = 1}^k x_{a_i} = 1 - d
	\qquad
	\text{and}
	\qquad
	\sum_{i = 1}^k x_{b_i}
	=
	dkh + (1 - d)(k - 1)
	=
	k(dh + 1 - d) + d - 1
	\enspace.
\]
Thus,
\begin{align*}
	\hat{G}_k(\vx)
	={} &
	\hat{F}_k(\bar{\vx})
	\leq
	F_k(\bar{\vx}) + \eps'
	=
	\sum_{1 = 1}^k \frac{\sum_{i = 1}^k x_{a_i}}{k}\left(1 - \frac{\sum_{i = 1}^k x_{b_i}}{k}\right) + \eps'\\
	={} &
	(1 - d)\left(d - dh + \frac{1 - d}{k}\right) + \eps'
	\leq
	d(1 - d)(1 - h) + \frac{1}{k} + \eps'
	\leq
	\frac{1 - h}{4} + \frac{3}{2k}
	\enspace.
	\qedhere
\end{align*}
\end{proof}

As mentioned above, we now would like to describe how the two basic instances are scrambled. Intuitively, the constraint polytope $\cK_{h,k,\ell}$ of a scrambled instance is obtained by combining $\ell$ orthogonal instances of $\cP_{h,k}$. Each element $a_i$ or $b_i$ has a copy in all the orthogonal instances, and the objective function treats every such copy as representing $\ell^{-1}$ of the original element. For example, if one would like to construct a solution assigning a value of $1/2$ to $a_i$, then the copies of $a_i$ in $\cK_{h,k,\ell}$ should get an average value of $1/2$. By randomly permuting the names of the elements in each orthogonal instance of $\cP_{h,k}$, we make it difficult for the algorithm to construct solutions that do not correspond to symmetric vectors in $\cP_{h,k}$. More formally, the constraint polytope $\cK_{h,k,\ell}$ is a subset of $[0, 1]^{\cM_{k,\ell}}$, where
\[
	\cM_{k,\ell}
	=
	\{a_{i, j}, b_{i, j} \mid i \in [k], j \in [\ell]\}
	\enspace.
\]
A vector $\vx \in [0, 1]^{\cM_{k, \ell}}$ belongs to $\cK_{h,k,\ell}$ if for every $j \in [\ell]$ we have $\vx^{(j)} \in \cP_{h, k}$, where the vector $\vx^{(j)} \in [0, 1]^{\cN_k}$ is defined by
\[
	\vx^{(j)}_{a_i} = \vx_{a_{i, j}}
	\qquad
	\text{and}
	\qquad
	\vx^{(j)}_{b_i} = \vx_{b_{i, j}}
	\enspace.
\]

The following lemma is an immediate corollary of the definition of $\cK_{h,k,\ell}$, Observation~\ref{obs:min_h} and the discussion before this observation.
\begin{lemma} \label{lem:min_h}
When $k \geq 1 / (1 - h)$, $\cK_{h,k,\ell}$ is solvable and $\max_{\vx \in \cK_{h, k,\ell}} \|\vx\|_\infty = h$.
\end{lemma}

The objective functions of the scrambled instances are formally defined using a vector $\vsigma$ of $\ell$ permutations over $[k]$ (in other words, $\sigma_1, \sigma_2, \dotsc, \sigma_\ell$ are all permutations over $[k]$). Given such a vector $\vsigma$ and a vector $\vx \in [0, 1]^{\cM_{k, \ell}}$, we define the vector
$\vx^{(\vsigma)} \in [0, 1]^{\cN_k}$ as follows.
\[
	\vx^{(\vsigma)}_{a_i}
	=
	\tfrac{1}{\ell} \sum_{j = 1}^\ell \vx_{a_{\sigma_j(i), j}}
	\qquad
	\text{and}
	\qquad
	\vx^{(\vsigma)}_{b_i}
	=
	\tfrac{1}{\ell} \sum_{j = 1}^\ell \vx_{b_{\sigma_j(i), j}}
	\enspace.
\]
Then, the functions $\bar{F}_{k, \vsigma}\colon [0, 1]^{\cM_{k, \ell}} \to \nnR$ and $\bar{G}_{k, \vsigma}\colon [0, 1]^{\cM_{k, \ell}} \to \nnR$ are defined for every vector $\vx \in [0, 1]^{\cM_{k, \ell}}$ by
\[
	\bar{F}_{k, \vsigma}(\vx)
	=
	\hat{F}(\vx^{(\vsigma)})
	\qquad
	\text{and}
	\qquad
	\bar{G}_{k, \vsigma}(\vx)
	=
	\hat{G}(\vx^{(\vsigma)})
	\enspace.
\]

The following lemma shows that the functions $\bar{F}_{k, \vsigma}$ and $\bar{G}_{k, \vsigma}$ inherit all the good properties of $\hat{F}_k$ and $\hat{G}_k$ promised by Lemma~\ref{lem:hat_functions_properties}. Since the proof of this lemma is technical and quite straightforward given Lemma~\ref{lem:hat_functions_properties}, we defer it to Appendix~\ref{app:missing_inapproximability}.
\begin{restatable}{lemma}{lemScrambledObjectives} \label{lem:scrambled_objectives}
The functions $\bar{F}_{k, \vsigma}$ and $\bar{G}_{k, \vsigma}$ are continuously differentiable, non-negative and DR-submodular. Furthermore, they are $\beta$-smooth for a value $\beta$ that is polynomial in $k$ and $\ell$.
\end{restatable}

We can now formally state the scrambled instances that we consider.
\[
	\begin{array}{l}
		\max \bar{F}_{k, \vsigma}(\vx)\\
		\vx \in \cK_{h, k, \ell}
	\end{array}
	\qquad
	\text{and}
	\qquad
	\begin{array}{l}
		\max \bar{G}_{k, \vsigma}(\vx)\\
		\vx \in \cK_{h, k, \ell}
	\end{array}
	\enspace.
\]
The next lemma shows that these scrambled instances inherit the values of their optimal solutions from the basic instances, which in particular, implies that they also inherit the gap between these solutions.
\begin{lemma} \label{lem:max_values_preserved}
We have both $\max_{\vx \in \cK_{h, k, \ell}} \bar{F}_{k, \vsigma}(\vx) = \max_{\vx \in \cP_{h, k}} \hat{F}_k(\vx)$ and $\max_{\vx \in \cK_{h, k, \ell}} \bar{G}_{k, \vsigma}(\vx) = \max_{\vx \in \cP_{h, k}} \hat{G}_k(\vx)$.
\end{lemma}
\begin{proof}
We prove below only the first equality of the lemma. The proof of the other equality is analogous. We begin by arguing that $\max_{\vx \in \cK_{h, k, \ell}} \bar{F}_{k, \vsigma}(\vx) \geq \max_{\vx \in \cP_{h, k}} \hat{F}_k(\vx)$. To show this inequality, we start with an arbitrary vector $\vx \in \cP_{h, k}$, and we construct a vector $\vy \in \cK_{h, k, \ell}$ such that $\bar{F}_{k, \vsigma}(\vy) = \hat{F}_k(\vx)$. Formally, the vector $\vy$ is defined as follows. For every $i \in [k]$ and $j \in [\ell]$,
\[
	\vy_{a_{i, j}} = \vx_{a_{\sigma^{-1}_j(i)}}
	\qquad
	\text{and}
	\qquad
	\vy_{b_{i, j}} = \vx_{b_{\sigma^{-1}_j(i)}}
	\enspace.
\]
One can observe that $\vx = \vy^{(\vsigma)}$, and therefore, we indeed have $\bar{F}_{k, \vsigma}(\vy) = \hat{F}_k(\vx)$; which means that we are only left to show that $\vy \in \cK_{h, k, \ell}$. Recall that, by the definition of $\cK_{h, k, \ell}$, to prove this inclusion, we need to argue that $\vy^{(j)} \in \cP_{h, k}$ for every $j \in [\ell]$, where $\vy^{(j)}$ is the restriction of $\vy$ to elements of $\{a_{i, j}, b_{i, j} \mid i \in [k]\}$.

Below, given a vector $\vz \in \cP_{h, k}$, we denote by $\sigma_j(\vz)$ the following vector.
\[
	(\sigma_j(\vz))_{a_i} = \vz_{a_{\sigma^{-1}_j(i)}}
	\qquad
	\text{and}
	\qquad
	(\sigma_j(\vz))_{b_i} = \vz_{b_{\sigma^{-1}_j(i)}}
	\enspace.
\]
Observe that this definition implies $\sigma_j(\vu) = \vu$ and $\sigma_j(\vv^{(i)}) = \vv^{(\sigma_j(i))}$, where $\vu, \vv^{(1)}, \vv^{(2)}, \dotsc,\allowbreak \vv^{(k)}$ are the vectors whose convex-hall defines $\cP_{h, k}$. Since $\vx \in \cP_{h, k}$, it must be given by some convex combination of the vectors $\vu, \vv^{(1)}, \vv^{(2)}, \dotsc, \vv^{(k)}$. In other words,
\[
	\vx = \sum_{i = 1}^k c_i \cdot \vv^{(i)} + d \cdot \vu
	\enspace.
\]
Thus, 
\[
	\vy^{(j)}
	=
	\sigma_j(\vx)
	=
	\sigma_j\left(\sum_{i = 1}^k c_i \cdot \vv^{(i)} + d \cdot \vu\right)
	=
	\sum_{i = 1}^k c_i \cdot \vv^{(\sigma_j(i))} + d \cdot \vu
	\enspace.
\]
The rightmost side of the last equality is another convex combination of the vectors $\vu, \vv^{(1)}, \vv^{(2)}, \dotsc,\allowbreak \vv^{(k)}$, and thus, the equality shows that $\vy^{(j)} \in \cP_{h, k}$, as desired.

We now get to the proof that $\max_{\vx \in \cK_{h, k, \ell}} \bar{F}_{k, \vsigma}(\vx) \leq \max_{\vx \in \cP_{h, k}} \hat{F}_k(\vx)$. Consider an arbitrary vector $\vx \in \cK_{h, k, \ell}$. By the definition of $\bar{F}_{k, \vsigma}(\vx)$, $\bar{F}_{k, \vsigma}(\vx) = \hat{F}_k(\vx^{(\vsigma)})$. Thus, to prove the last inequality, it suffices to show that $\vx^{(\vsigma)} \in \cP_{h, k}$, which is done by the next claim. Since the proof of this claim is very similar to the above proof that $\vy \in \cK_{h, k, \ell}$, we defer it to Appendix~\ref{app:missing_inapproximability}.
\begin{restatable}{claim}{clmXSigmaP} \label{clm:x_sigma_P}
For every vector $\vx \in \cK_{h, k, \ell}$, $\vx^{(\vsigma)} \in \cP_{h, k}$.\iftoggle{proofsAppendix}{}{\qedhere}
\end{restatable}
\end{proof}

\begin{corollary} \label{cor:scrambled_opt}
It holds that $\max_{\vx \in \cK_{h, k, \ell}} \bar{F}_{k, \vsigma}(\vx) \geq 1 - 1 / (2k) = \Omega(k^{-1})$ and $(1 - h)/4 + 3/(2k) \geq \max_{\vx \in \cK_{h, k, \ell}} \bar{G}_{k, \vsigma}(\vx) = \Omega(k^{-1})$.
\end{corollary}

Lemmata~\ref{lem:min_h}, \ref{lem:scrambled_objectives} and~\ref{lem:max_values_preserved} show that the scrambled instances we have constructed have all the properties stated in Theorem~\ref{thm:inapproximability} when $k \geq 1 / (1 - h)$). Therefore, to prove the theorem it suffices to show that no sub-exponential time algorithm can obtain a good approximation guarantee given these instances when $\ell$ is large enough compared to $k$. We do this by showing that when $\vsigma$ is chosen uniformly at random, it is difficult to distinguish between the two scrambled instances, and therefore, no sub-exponential time algorithm can obtain an approximation ratio better than the (large) gap between their optimal values. The first step in this proof is done by the next lemma, which shows that any single access to the objective function almost always returns the same answer given either of the two scrambled instances. To understand why the lemma implies this, it is important to understand what we mean with an access to the function. In this paper, we assume the ability to access either the objective, or its gradient, at a point $\vx$. The answers for both these types of accesses are determined by the values of the objective at an arbitrarily small neighborhood of $\vx$, and the same is true also for many other natural kinds of access (such as higher order derivatives).
\begin{lemma} \label{lem:single_access}
Assume $\vsigma$ is drawn uniformly at random, i.e., $\sigma_j$ is an independently chosen uniformly random permutation of $[k]$ for every $j \in [\ell]$. Given any vector $\vx \in [0, 1]^{\cM_k}$, with probability at least $1 - 4k \cdot e^{-\ell \cdot \frac{\delta_k}{6\sqrt{2k}}}$ we have $\bar{F}_{k, \vsigma}(\vy) = \bar{G}_{k, \vsigma}(\vy)$ for every vector $\vy$ such that $\|\vx - \vy\|_2 \leq (\sqrt{\ell} / 4) \cdot \delta_k$, where $\delta_k$ is the value of $\delta$ when Lemma~\ref{lem:original_continuous_versions} is applied to $f_k$.
\end{lemma}
\begin{proof}
Below, we show that $\|\vx^{(\vsigma)} - \bar{\vx}^{(\vsigma)}\|_2 \leq \delta_k/2$ with probability at least $1 - 4k \cdot e^{-\ell \delta_k / (6\sqrt{2k})}$. However, before getting to this proof, let us show that, whenever this inequality holds, we also have $\bar{F}_{k, \vsigma}(\vy) = \bar{G}_{k, \vsigma}(\vy)$. By the definitions of $\bar{F}_{k, \vsigma}$ and $\bar{G}_{k, \vsigma}$, the last equality is equivalent to $\hat{F}_k(\vy^{(\vsigma)}) = \hat{G}_k(\vy^{(\vsigma)})$, and this equality holds by Lemma~\ref{lem:original_continuous_versions} since
\[
	\|\vy^{(\vsigma)} - \bar{\vy}^{(\vsigma)}\|_2
	\leq
	\|\vy^{(\vsigma)} - \vx^{(\vsigma)}\|_2 + \|\bar{\vy}^{(\vsigma)} - \bar{\vx}^{(\vsigma)}\|_2 + \|\vx^{(\vsigma)} - \bar{\vx}^{(\vsigma)}\|_2
	\leq
	2\|\vy^{(\vsigma)} - \vx^{(\vsigma)}\|_2 + \delta_k/2
	\leq
	\delta_k
	\enspace,
\]
where the first inequality is the triangle inequality, the second inequality holds since averaging two vectors in the same way can only decrease their distance from each other, and the last inequality holds because Sedrakyan's inequality (or Cauchy–Schwarz inequality) implies
\begin{align*}
	\|\vy^{(\vsigma)} - \vx^{(\vsigma)}\|_2^2
	={} &
	\frac{\sum_{i = 1}^k [\sum_{j = 1}^\ell (\vy_{a_{\sigma_j(i), j}} - \vx_{a_{\sigma_j(i), j}})]^2 + \sum_{i = 1}^k [\sum_{j = 1}^\ell (\vy_{b_{\sigma_j(i), j}} - \vx_{b_{\sigma_j(i), j}})]^2}{\ell^2}\\
	\leq{} &
	\frac{\sum_{i = 1}^k \sum_{j = 1}^\ell (\vy_{a_{\sigma_j(i), j}} - \vx_{a_{\sigma_j(i), j}})^2 + \sum_{i = 1}^k \sum_{j = 1}^\ell (\vy_{b_{\sigma_j(i), j}} - \vx_{b_{\sigma_j(i), j}})^2}{\ell}
	=
	\frac{\|\vx - \vy\|_2^2}{\ell}
	\enspace.
\end{align*}

It now remains to prove that the inequality $\|\vx^{(\vsigma)} - \bar{\vx}^{(\vsigma)}\|_2 \leq \delta_k/2$ holds with probability at least $1 - 4k \cdot e^{-\ell \delta_k / (6\sqrt{2k})}$. By the union bound, to prove this inequality it suffices to show that, for every $i \in [k]$, the probabilities of the two inequalities $|\vx^{(\vsigma)}_{a_i} - \bar{\vx}^{(\vsigma)}_{a_i}| > \delta_k/\sqrt{8k}$ and $|\vx^{(\vsigma)}_{b_i} - \bar{\vx}^{(\vsigma)}_{b_i}| > \delta_k/\sqrt{8k}$ to hold are both at most $2e^{-\ell \delta_k / (6\sqrt{2k})}$. The rest of this proof is devoted to showing that this is indeed the case for the first inequality as the proof for the second inequality is analogous. Recall that
\begin{equation} \label{eq:vx_def}
	\vx^{(\vsigma)}_{a_i}
	=
	\tfrac{1}{\ell} \sum_{j = 1}^\ell \vx_{a_{\sigma_j(i), j}}
	\enspace.
\end{equation}
Thus,
\begin{equation} \label{eq:bar_independent_of_sigma}
	\bar{\vx}^{(\vsigma)}_{a_i}
	=
	\frac{1}{k} \sum_{i' = 1}^k \vx^{(\vsigma)}_{a_{i'}}
	=
	\frac{1}{k} \sum_{i' = 1}^k \left(\tfrac{1}{\ell} \sum_{j = 1}^\ell \vx_{a_{\sigma_j(i'), j}}\right)
	=
	\frac{1}{k\ell} \sum_{i' = 1}^k \sum_{j = 1}^\ell \vx_{a_{\sigma_j(i'), j}}
	=
	\frac{1}{k\ell} \sum_{i' = 1}^k \sum_{j = 1}^\ell \vx_{a_{i', j}}
	\enspace,
\end{equation}
where the last equality holds since $\sigma_j$ is a permutation over $[k]$. Similarly, we also have
\[
	\bE[\vx^{(\vsigma)}_{a_i}]
	=
	\frac{1}{\ell} \sum_{j = 1}^\ell \bE[\vx_{a_{\sigma_j(i), j}}]
	=
	\frac{1}{\ell} \sum_{j = 1}^\ell \left(\tfrac{1}{k} \sum_{i' = 1}^k \bE[\vx_{a_{i', j}}]\right)
	=
	\bar{\vx}^{(\vsigma)}_{a_i}
	\enspace.
\]
Hence, the claim that we want to prove bounds the probability that $\vx^{(\vsigma)}_{a_i}$ significantly deviates from its expectation. Furthermore, Equation~\eqref{eq:vx_def} shows that $\ell \cdot \vx^{(\vsigma)}_{a_i}$ is the sum of $\ell$ random variables taking values from the range $[0, 1]$. Since $\sigma_j$ is chosen independently for every $j \in [\ell]$, these $\ell$ random variables are independent, which allows us to use Chernoff's inequality to bound their sum. Therefore,
\begin{align*}
	\Pr\left[|\vx^{(\vsigma)}_{a_i} - \bar{\vx}^{(\vsigma)}_{a_i}| > \frac{\delta_k}{\sqrt{8k}}\right]
	&=
	\Pr\left[\left|\sum_{j = 1}^\ell \vx_{a_{\sigma_j(i), j}} - \bE\left[\sum_{j = 1}^\ell \vx_{a_{\sigma_j(i), j}}\right]\right| > \frac{\ell \delta_k}{\sqrt{8k}}\right]\\
	\leq{} &
	2e^{-\frac{\bE[\sum_{j = 1}^\ell \vx_{a_{\sigma_j(i), j}}] \cdot \min\left\{\frac{\ell \delta_k}{\sqrt{8k} \cdot \bE[\sum_{j = 1}^\ell \vx_{a_{\sigma_j(i), j}}]}, \frac{\ell^2 \delta^2_k}{8k \cdot \bE[\sum_{j = 1}^\ell \vx_{a_{\sigma_j(i), j}}]^2}\right\}}{3}}\\
	={} &
	2e^{-\frac{\min\left\{\frac{\ell \delta_k}{\sqrt{8k}}, \frac{\ell^2 \delta^2_k}{8k \cdot \bE[\sum_{j = 1}^\ell \vx_{a_{\sigma_j(i), j}}]}\right\}}{3}}
	\leq
	2e^{-\frac{\frac{\ell \delta_k}{\sqrt{8k}} \cdot \min\left\{1, \frac{\delta_k}{\sqrt{8k}}\right\}}{3}}
	=
	2e^{-\ell \cdot \frac{\delta_k}{6\sqrt{2k}}}
	\enspace.
	\qedhere
\end{align*}
\end{proof}

Equation~\eqref{eq:bar_independent_of_sigma} in the last proof has another interesting consequence. This equation shows that $\bar{\vx}^{(\vsigma)}$ is independent of $\vsigma$. Since Lemma~\ref{lem:original_continuous_versions} shows that $\hat{G}_k(\vx) = \hat{F}_k(\bar{\vx})$ for every $\vx \in [0, 1]^{\cN_k}$, this implies the following observation.
\begin{observation}
For every $\vx \in [0, 1]^{\cM_{k, \ell}}$, the value of $\bar{G}_{k, \vsigma}(\vx) = \hat{G}_k(\vx^{(\vsigma)}) = \hat{F}_k(\bar{\vx}^{\vsigma})$ is independent of $\vsigma$.
\end{observation}

In light of the above observation, we use below $\bar{G}_k$ to denote the function $\bar{G}_{k, \vsigma}$. We are now ready to prove Theorem~\ref{thm:inapproximability}.
\begin{proof}[Proof of Theorem~\ref{thm:inapproximability}]
Fix an arbitrary sub-exponential function $P(\cdot)$. Below, we show that there is a distribution of instances on which no deterministic algorithm making at most $P(n)$ accesses to the objective function, where $n$ is the dimension, can obtain an approximation ratio of $(1 - h)/4 + \varepsilon$. By Yao's principle, this will imply the same result also for randomized algorithms running in time $P(n)$ (notice that running in time $P(n)$ implies making at most $P(n)$ accesses to the objective function).

The distribution of instances we consider is the scrambled instance $\max_{vx \in \cK_{h, k, \ell}} \cF_{k, \vsigma}$, where $k \geq 1 / (1 - h)$ and $\ell$ are deterministic values to be determined below, and $\vsigma$ is chosen at random according to the distribution defined in Lemma~\ref{lem:single_access}. Assume towards a contradiction that there exists a deterministic algorithm $ALG$ that accesses the objective function at most $P(|\cM_{k, \ell}|) = P(2k\ell)$ times, and given a random instance from the above distribution obtains an approximation ratio of $(1 - h)/4 + \varepsilon$. More formally, if we denote $OPT = \max_{\vx \in \cP_{h,k}} \hat{F}_k(\vx)$, then $ALG$ guarantees that its output vector $\va$ obeys
\begin{equation} \label{eq:contradicted_assumption}
	\bE[\cF_{k, \vsigma}(\va)]
	\geq
	[(1 - h)/4 + \varepsilon] \cdot \bE\left[\max_{\vx \in \cK_{h,k, \ell}} \hat{F}_{k, \vsigma}(\vx)\right]
	=
	[(1 - h)/4 + \varepsilon] \cdot OPT
	\enspace,
\end{equation}
where the equality holds by Lemma~\ref{lem:max_values_preserved}.

Consider now an execution of $ALG$ on the instance $\max_{\vx \in \cK_{h, k, \ell}} \bar{G}_k(\vx)$, and let us denote by $A_1, A_2, \dotsc, A_r$ the accesses made by $ALG$ (each access $A_i$ consists of a vector $\vx$ and the type of access, namely whether $ALG$ evaluates the objective function at $\vx$ or calculates the gradient of the objective function at $\vx$). It is convenient to assume that the last access made by $ALG$ is to evaluate the value of its output set $\va$. If this is not the case, we can add such an access to the end of the execution of $ALG$, and still have $r \leq P(2k\ell) + 1$. Let $\cE$ be the event that all the accesses $A_1, A_2, \dotsc, A_r$ return the same value given that the objective is either $\bar{G}_k$ or $\bar{F}_{k, \vsigma}$. Clearly, $ALG$ follows the same execution path given either $\bar{G}_k$ or $\bar{F}_{k, \vsigma}$ when the event $\cE$ happens, and therefore, it outputs the same vector $\va \in \cK_{h, k, \ell}$ in this case. Furthermore, $\cE$ also implies that $\bar{F}_{k, \vsigma}(\va) = \bar{G}_k(\va)$, and thus, conditioned on $\cE$,
\begin{align*}
	\cF_{k, \vsigma}(\va)
	\leq{} &
	\max_{\vx \in \cK_{h,k, \ell}} \hat{G}_k(\vx)
	\leq
	(1-h)/4+3/(2k)
	\leq
	\frac{(1-h)/4+3/(2k)}{1 - 1/(2k)} \cdot OPT\\
	\leq{} &
	\left[\frac{1 - h}{4 - 2/k} + \frac{3}{k}\right] \cdot OPT
	\leq
	\left[\frac{1 - h}{4} + \frac{4}{k}\right] \cdot OPT
	\enspace,
\end{align*}
where the second inequality holds by Corollary~\ref{cor:scrambled_opt}, the third inequality follows from Lemma~\ref{lem:opt_bounds}, and two last inequalities hold since $k \geq 1$ and $h \in [0, 1]$.

We would like to use the last inequality to upper bound $\bE[\cF_{k, \vsigma}(\va)]$. For that purpose, we need to lower bound the probability of the event $\cE$. 
By Lemma~\ref{lem:single_access} and the union bound,
\[
	\Pr[\cE] \geq 1 - 4kr \cdot e^{-\ell \cdot \frac{\delta_k}{6\sqrt{2k}}} \geq 1 - 4k[P(2k \ell) + 1] \cdot e^{-\ell \cdot \frac{\delta_k}{6\sqrt{2k}}}
	\enspace.
\]
Consider the second term in the rightmost side of the last inequality. This term is a function of $k$ and $\ell$ alone, and for a fixed value of $k$ it is the product of a sub-exponential function of $\ell$ and an exponentially decreasing function of $\ell$. Therefore, for any fixed value of $k$, we can choose a large enough value for $\ell$ to guarantee that $2k[P(2k \ell) + 1] \cdot e^{-\ell \cdot \frac{\delta_k}{6\sqrt{k}}} \leq \eps/2$. In the rest of the proof we assume that $\ell$ is chosen in such a way. Then, since we always have $\cF_{k, \vsigma}(\va) \leq OPT$ and $\Pr[\cE] \leq 1$, we get by the law of total expectation,
\[
	\bE[\cF_{k, \vsigma}(\va)]
	\leq
	\Pr[\bar{\cE}] \cdot OPT + \bE[\cF_{k, \vsigma}(\va) \mid \cE]
	\leq
	\frac{\eps}{2} \cdot OPT + [(1 - h)/4 + 4/k] \cdot OPT
	\enspace,
\]
which contradicts Equation~\eqref{eq:contradicted_assumption} (and thus, the existence of $ALG$) when $k$ is chosen to be $\max\{\lceil 1 / (h - 1) \rceil, 8/\eps\}$.
\end{proof}
\section{Applications and Experimental Results}\label{sec:Experiments}

Up until recently, all the algorithms suggested for submodular maximization subject to general convex set constraints had a sub-exponential execution time. As mentioned above, Du~\cite{du2022lyapunov} has recently shown the first polynomial time offline algorithm for this problem, and in this paper we have shown another polynomial time algorithm obtaining a similar guarantee for the online (regret minimization) setting. In this section, we study the empirical performance of both these algorithms on multiple machine learning applications. In the case of the offline algorithm, it is important to note that (i) we analyze our explicit version of the algorithm, rather than the original version of Du~\cite{du2022lyapunov}; and (ii) it is interesting to study the empirical performance of the algorithm of Du~\cite{du2022lyapunov} because only a theoretical analysis of this algorithm appeared in~\cite{du2022lyapunov}.

Since the previously suggested algorithms require sub-exponential execution time, and thus cannot be used as is, we allowed all algorithms in our experiments the same number of iterations. This makes all the algorithms terminate in roughly the same amount of time, and allows for a fair comparison between the quality of their solutions. In a nutshell, our experiments show that our online algorithm and the offline algorithm of Du~\cite{du2022lyapunov} provide better solutions (often much better) compared to their state-of-the-art sub-exponential time counterparts.

\subsection{Revenue Maximization}

Following~\cite{thang2021online}, our first set of experiments considers revenue maximization in the following setting. The goal of a company is to advertise a product to users so that the revenue increases through the ``word-of-mouth" effect. Formally, the input for the problem is a weighted undirected graph $G = (V, E)$ representing a social network graph, where $w_{ij}$ denotes the weight of the edge between vertex $i$ and vertex $j$ ($w_{ij} = 0$ if the edge $(i, j)$ is missing from the graph). If the company invests $x_i$ unit of cost in a user $i\in V$, then this user becomes an advocate of the product with probability $1-(1-p)^{x_i}$, where $p\in(0,1)$ is a parameter. Note that this means that each $\eps$ unit of cost invested in the user has an independent chance to make the user an advocate, and that by investing a full unit in the user, she becomes an advocate with probability $p$~\cite{soma2017nonmonotone}.

Let $S\subseteq V$ be a set of users who ended up being advocates for the product. Then, the revenue obtained is represented by the total influence of the users of $S$ on non-advociate users, or more formally, by $\sum_{i\in S}\sum_{j\in{V\setminus S}}w_{ij}$. The objective function $f\colon [0, 1]^V \rightarrow \nnR$ of the experiments is accordingly defined as the expectation of the above expression, i.e., 
\begin{align}
    f(\vx)=\mathbb{E}_S\left[\sum_{i\in S}\sum_{j\in{V\setminus S}}w_{ij}\right]=\sum_{i \in V}\sum_{\substack{j \in V \\ i\neq j}}w_{ij}(1-(1-p)^{x_i})(1-p)^{x_j}
		\enspace.
\end{align}
It has been shown that $f$ is a non-monotone DR-submodular function~\cite{soma2017nonmonotone}.

In both the online and offline settings, we experimented on instances of the above setting based on two different datasets. The first dataset is a Facebook network~\cite{viswanath2009evolution}, and includes $64K$ users (vertices) and $1M$ unweighted relationships (edges). The second dataset is based on the Advogato network~\cite{massa2009dowling}, and includes $6.5K$ users (vertices) as well as $61K$ weighted relationships (edges).

\subsubsection{Online setting} \label{ssc:revenue_online}

When performing our experiments in the online settings, we tried to closely mimic the experiment of~\cite{thang2021online}. Therefore, we chose the number of time steps to be $T = 1000$, and the parameter $p=0.0001$. In each time step $t$, the objective function is defined in the following way. A subset $V^t \subseteq V$ is selected, and only edges connecting two vertices of $V^t$ are kept. In the case of the Advogato network, $V_t$ is a uniformly random subset of $V$ of size $200$, and in the case of the much larger Facebook network, $V_t$ is a uniformly random subset of $V$ of size $\numprint{15,000}$. The optimization is done subject to the constraint $0.1\leq\sum_ix_i\leq 1$, which represents both minimum and maximum investment requirements. Note that the intersection of this constraint with the implicit box constraint represents a non-down-monotone feasibility polytope.

\begin{figure}[tb]
\begin{subfigure}[t]{0.24\textwidth}
  \begin{tikzpicture}[scale=0.45] \begin{axis}[
    xlabel = {Time Step},
    ylabel = {Function Value},
    xmin=0, xmax=1000,
    ymin=0, ymax=4,
		legend cell align=left,
		legend style={at={(0.8,1)}}]
		\addplot [name path = our, blue, mark = *, mark repeat=100] table [x expr=\coordindex, y index=0] {CSV/revenue/OnlineAdvOurs.csv};
		\addlegendentry{Our Algorithm}
		\addplot [name path = theirs, red, mark = triangle*, mark size=3pt, mark repeat=100] table [x expr=\coordindex, y index=0] {CSV/revenue/OnlineAdvTheirs.csv};
		\addlegendentry{Th\twodias{\'}{\u}{a}ng and Srivastav~\cite{thang2021online}}
	\end{axis}\end{tikzpicture}
  \caption{\centering Online Algorithms on the Advogato network.}\label{fig:onAdv}
\end{subfigure}\hfill
\begin{subfigure}[t]{0.24\textwidth}
    \begin{tikzpicture}[scale=0.45] \begin{axis}[
    xlabel = {Time Step},
    ylabel = {Function Value},
    xmin=0, xmax=1000,
    ymin=0, ymax=5,
		legend cell align=left,
		legend style={at={(0.8,1)}}]
		\addplot [name path = our, blue, mark = *, mark repeat=100] table [x expr=\coordindex, y index=0] {CSV/revenue/FB_online_ours.csv};
		\addlegendentry{Our Algorithm}
		\addplot [name path = theirs, red, mark = triangle*, mark size=3pt, mark repeat=100] table [x expr=\coordindex, y index=0] {CSV/revenue/FB_online_theirs.csv};
		\addlegendentry{Th\twodias{\'}{\u}{a}ng and Srivastav~\cite{thang2021online}}
	\end{axis}\end{tikzpicture}
  \caption{\centering Online Algorithms on the Facebook network.}\label{fig:onFace}
\end{subfigure}\hfill
\begin{subfigure}[t]{0.25\textwidth}
    \begin{tikzpicture}[scale=0.45] \begin{axis}[
    xlabel = {Iterations},
    ylabel = {Function Value},
    xmin=0, xmax=100,
    ymin=0, ymax=0.11,
		legend cell align=left,
		legend style={at={(0.68,1)}},
		y tick label style={
        /pgf/number format/.cd,
        fixed,
        fixed zerofill,
        precision=2,
        /tikz/.cd
    }]
		\addplot [name path = our, blue, mark = *, mark repeat=10] table [x expr=\coordindex, y index=0] {CSV/revenue/OfflineAdvogato0urs.csv};
		\addlegendentry{Du~\cite{du2022lyapunov} (our version)}
		\addplot [name path = theirs, red, mark = triangle*, mark size=3pt, mark repeat=10] table [x expr=\coordindex, y index=0] {CSV/revenue/OfflineAdvogatoTheirs.csv};
		\addlegendentry{D{\"u}rr et al.~\cite{durr2021non}}
		\addplot [name path = theirs, green!80!black, mark = square*, mark repeat=10] table [x expr=\coordindex, y index=0] {CSV/revenue/OfflineAdvogatoTheirs2.csv};
		\addlegendentry{Du et al.~\cite{du2022improved}}
	\end{axis}\end{tikzpicture}
  \caption{\centering Offline Algorithms on the Advogato network.}\label{fig:offAdv}
\end{subfigure}\hfill
\begin{subfigure}[t]{0.25\textwidth}
    \begin{tikzpicture}[scale=0.45] \begin{axis}[
    xlabel = {Iterations},
    ylabel = {Function Value},
    xmin=0, xmax=100,
    ymin=0, ymax=0.15,
		legend cell align=left,
		legend style={at={(0.68,1)}},
		y tick label style={
        /pgf/number format/.cd,
        fixed,
        fixed zerofill,
        precision=2,
        /tikz/.cd
    }]
		\addplot [name path = our, blue, mark = *, mark repeat=10] table [x expr=\coordindex, y index=0] {CSV/revenue/OfflineFB0urs.csv};
		\addlegendentry{Du~\cite{du2022lyapunov}  (our version)}
		\addplot [name path = theirs, red, mark = triangle*, mark size=3pt, mark repeat=10] table [x expr=\coordindex, y index=0] {CSV/revenue/OfflineFBTheirs.csv};
		\addlegendentry{D{\"u}rr et al.~\cite{durr2021non}}
		\addplot [name path = theirs, green!80!black, mark = square*, mark repeat=10] table [x expr=\coordindex, y index=0] {CSV/revenue/OfflineFBTheirs2.csv};
		\addlegendentry{Du et al.~\cite{du2022improved}}
	\end{axis}\end{tikzpicture}
  \caption{\centering Offline Algorithms on the Facebook network.}\label{fig:offFace}
\end{subfigure}
\caption{Results of the Revenue Maximization Experiments} \label{fig:revenue}
\end{figure}
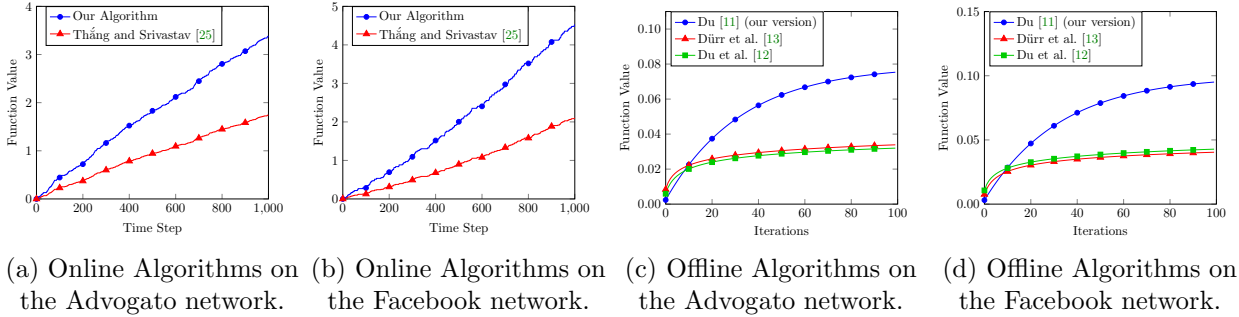

In our experiments, we have compared our algorithm from Section~\ref{sec:Online} with the algorithm of Th\twodias{\'}{\u}{a}ng and Srivastav~\cite{thang2021online}, which is the only other algorithm for the online setting currently known. In both algorithms, we have set the number of online linear optimizers used to be $L = 100$, and in our algorithm we have set the error parameter $\eps = 0.03$ (there is no error parameter in the algorithm of Th\twodias{\'}{\u}{a}ng and Srivastav~\cite{thang2021online}). The results of these experiments on the Advogato and Facebook networks can be found in Figures~\ref{fig:onAdv} and~\ref{fig:onFace}, respectively. One can observe that our algorithm significantly outperforms the state-of-the-art algorithm for any number of time steps.

\subsubsection{Offline setting}

Our experiments in the offline setting are similar to the ones done in the online setting, with two differences. First, since there is only one objective function in the offline setting, we base it on the entire network graph rather than on a subset of its vertices. Second, for the sake of diversity, we changed the constraint to be $0.25\leq\sum_ix_i\leq 1$ (but we note that the results of the experiments remain essentially unchanged if one reuse the constraint from the online setting).

In our experiments, we have compared our explicit version from Section~\ref{sec:Offline} of the algorithm of Du~\cite{du2022lyapunov} with the previous algorithms of D{\"u}rr et al.~\cite{durr2021non} and Du et al.~\cite{du2022improved}. All the algorithms have been executed for $T = 100$ iterations,\footnote{Recall that the number of iterations corresponds to the parameter $L$ in the online setting, which was also set to $100$ above.} and the error parameter $\eps$ was set $0.03$ in (our version of) the algorithm of Du~\cite{du2022lyapunov}. The results of these experiments on the Advogato and Facebook networks can be found in Figures~\ref{fig:offAdv} and~\ref{fig:offFace}, respectively. One can observe that our version of the polynomial time algorithm of Du~\cite{du2022lyapunov} clearly outperforms the two previous algorithms, except when the number of iterations is very low.

\subsection{Location Summarization}

In this section we consider a location summarization task based on the Yelp dataset~\cite{yelp}, which is a subset of Yelp’s businesses, reviews and user data. This dataset contains information about local businesses across $11$ metropolitan areas, and we have followed the technique of~\cite{kazemi2021regularized} for generating symmetry scores between these locations based on features extracted from the descriptions of the locations and their related user reviews (such as parking options, WiFi access, having vegan menus, delivery options, possibility of outdoor seating and being good for groups).

We would like to pick a non-empty set of up to $2$ locations that summarizes the existing locations, while not being too far from the current location of the user. A natural objective function for this task (which is very similar to the objective function used in~\cite{kazemi2021regularized}) is the following set function. Assume that the set of locations is $[n]$, $M_{i,j}$ is the similarity score between locations $i$ and $j$, and $d_i$ is the distance of location $i$ from the user (in units of 200KM); then for every set $S \subseteq [n]$, the value of the objective is
\[
    f(S) = \tfrac{1}{n} \sum_{i = 1}^n \max_{j \in S} M_{i, j} - \sum_{i \in S} d_i
    \enspace.
\]

Since $f$ is a set function, and the tools we have developed in this work apply only to continuous functions, we optimize the multilinear extension $F$ of $f$,\footnote{See Section~\ref{sec:Hardness} for a formal definition of the multi-linear extension.} which is given for every vector $\vx \in [0, 1]^n$ by 
\[
    F(\vx)
    =
    \tfrac{1}{n} \sum_{i = 1}^n \sum_{j = 1}^n \left[x_j M_{i, j} \cdot \prod_{j' | M_{i,j} \prec M_{i, j'}} \mspace{-27mu} (1 - x_{j'}) \right] - \sum_{i = 1}^n x_i d_i
    \enspace.
\]
The multilinear extension $F$ is DR-submodular since $f$ is submodular. Furthermore, any solution obtained while optimizing $F$ can be rounded into a solution obtaining the same approximation guarantee for $f$ using either pipage or swap rounding~\cite{calinescu2011maximizing,chekuri2010dependent}.

In our experiment, we restricted attention to a single metropolitan area (Charlotte), and assumed there are $100$ time steps. In each time step, a new user $u$ arrives, and her location is determined uniformly at random within the rectangle containing the metropolitan area. Let us denote by $F_u$ the function $F$ when the distances are calculated based on the location of $u$. When user $u$ arrives, we would like to choose a vector $\vx^{(u)}$ maximizing $F_u$ among all vectors obeying $\|\vx\|_1 \in [1, 2]$ (recall that we look for solutions that include $1$ or $2$ locations). Furthermore, we would like to do that before learning the location of $u$ (to speed up the response and for privacy reasons); thus, we need to consider online optimization algorithms. Specifically, like in Section~\ref{ssc:revenue_online}, we compared our algorithm from Section~\ref{sec:Online} with the algorithm of Th\twodias{\'}{\u}{a}ng and Srivastav~\cite{thang2021online}. In both algorithms, we have set the number of online linear optimizers used to be $L = 100$, and in our algorithm we have set the error parameter $\eps = 0.03$. The results of the experiment can be found in Figure~\ref{fig:location}, and they show that our algorithm (again) significantly outperforms the state-of-the-art algorithm for any number of time steps.


\begin{SCfigure}[][tb]
  \begin{tikzpicture}[scale=0.8] \begin{axis}[
    xlabel = {Time Step},
    ylabel = {Function Value},
    xmin=0, xmax=100,
    ymin=0, ymax=110,
		legend cell align=left,
		legend style={at={(0.8,1)}}]
		\addplot [name path = our, blue, mark = *, mark repeat=10] table [x expr=\coordindex, y index=0] {CSV/location/YelpCharlotte_ours12.csv};
		\addlegendentry{Our Algorithm}
		\addplot [name path = theirs, red, mark = triangle*, mark size=3pt, mark repeat=10] table [x expr=\coordindex, y index=0] {CSV/location/YelpCharlotte_theirs12.csv};
		\addlegendentry{Th\twodias{\'}{\u}{a}ng and Srivastav~\cite{thang2021online}}
	\end{axis}\end{tikzpicture}
\caption{Results of the Location Summarization Experiment} \label{fig:location}
\end{SCfigure}
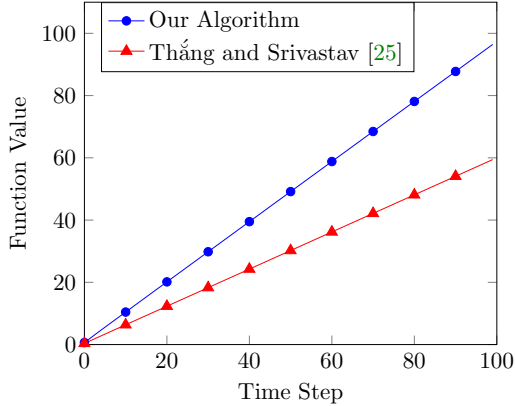

\subsection{Quadratic programming}

In this section, we complement the study of (our version) of the offline algorithm of Du~\cite{du2022lyapunov}, by checking its empirical performance for down-closed polytopes. Algorithms with better approximation guarantees are known when one is guaranteed to have such a constraint~\cite{bian2017nonmonotone}. However, it is still important to understand the performance of algorithms designed for general polytope constraint when they happen to get a down-closed polytope. In particular, we note that D{\"u}rr et al.~\cite{durr2021non} studied the empirical performance of their algorithm compared to the performance of the algorithm of~\cite{bian2017nonmonotone} subject to such constraints, and we extend here their work by comparing the performance of their algorithm with that of newer algorithms. All the experiments presented in this section closely follow settings studied in~\cite{durr2021non}.

Consider the down-closed polytope given by
\[
    \mathcal{K}=\{\vx\in \nnRE{n} \mid \mA\vx\leq \vb, \vx\leq \vu,\mA\in\nnRE{m\times n}, \vb\in \nnRE{m}\}
		\enspace,
\]
where $\mA$ is a non-negative matrix chosen in a way described below, $\vb$ is the all ones vector, and $\vu$ is a vector that acts as an upper bound on $\cK$ and is given by $u_j=\min_{j\in[m]} b_i / A_{i,j}$ for every $j\in[n]$. We now describe a function $F$ that we would like to maximize subject to $\cK$. For every vector $\vzero \leq \vx \leq \vu$ (where $\vzero$ is the all zeros vector),
\[
    F(\vx)=\frac{1}{2}\vx^T\mH\vx+\vh^T\vx+c
		\enspace,
\]
where $\mH$ is a matrix, $\vh$ is a vector and $c$ is a scalar. The matrix $\mH$ is chosen in a way described below, and it is always non-positive, which guarantees that $F$ is DR-submodular. Furthermore, once $\mH$ is chosen, we follow~\cite{bian2017nonmonotone} and set $\vh = -0.1\cdot \mH^T\vu$. Finally, to make sure that $F$ is also non-negative, the value of $c$ should be at least $M = -\min_{\vzero \leq \vx \leq \vu}\frac{1}{2}\vx^T\mH\vx + \vh^T\vx$ . The value of $M$ can be approximately obtained using \quadprogIP\footnote{We used IBM CPLEX optimization studio \url{https://www.ibm.com/products/ilog-cplex-optimization-studio}.}~\cite{xia2020globally}, and $c$ is chosen to be $M+0.1|M|$, which is a bit larger than the necessary minimum.

It remains to describe the way in which the entries of the matrices $\mH$ and $\mA$ are chosen. Below we describe two different random ways in which this can be done, and study the performance of the various algorithms on the instances generated in this way.

\subsubsection{Uniform distribution} The first way to choose the matrices $\mH$ and $\mA$ is using a uniform distribution. Here, the matrix $\mH\in \mathbb{R}^{n\times n}$ is a randomly generated symmetric matrix whose entries are drawn uniformly at random (and independently) from $[-1,0]$, and $\mA\in\mathbb{R}^{m\times n}$ is a randomly generated matrix whose entries are drawn uniformly at random from $[v,v+1]$ for $v=0.01$ (this choice of $v$ guarantees that the entries of $\mA$ are strictly positive).

In each one of our experiments, we chose a different set of values for the dimensions $n$ and $m$, and then drew an instance from the above distribution and executed on it $100$ iterations of three algorithms: our explicit version from Section~\ref{sec:Offline} of the algorithm of Du~\cite{du2022lyapunov} (with $\eps = 0.03$), and the previous algorithms of D{\"u}rr et al.~\cite{durr2021non} and Du et al.~\cite{du2022improved}. Each such experiment was repeated $100$ times, and the results are depicted in Figure~\ref{fig:quad11}. In each plot of this figure, the $x$-axis represents the value of $n$, and the caption of the plot specifies how the value of $m$ was calculated based on the value of $n$. The $y$-axis of the plots represents the approximation ratios obtained by the various algorithms compared to the optimum computed using a quadratic programming solver. One can observe that the two sub-exponential time algorithms of D{\"u}rr et al.~\cite{durr2021non} and Du et al.~\cite{du2022improved} exhibit similar performance, and (our version) of the newer algorithm of Du~\cite{du2022lyapunov} consistently and significantly outperforms them.

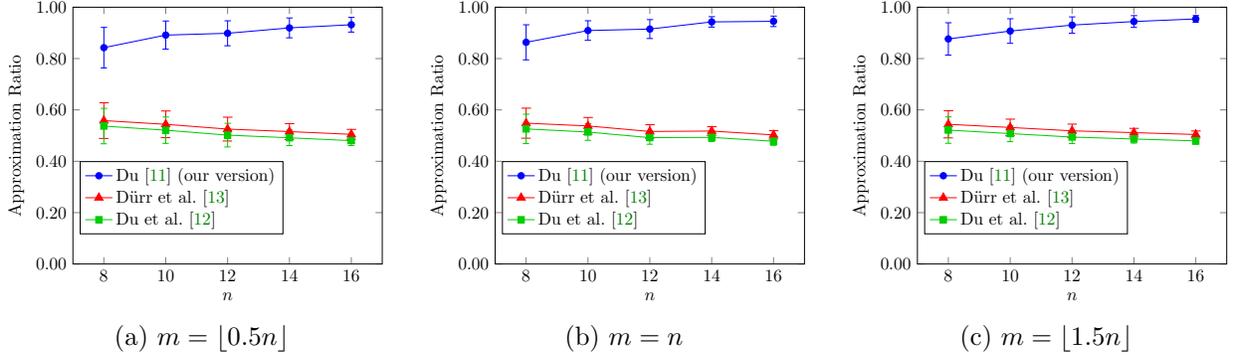
\begin{figure}[tb]
\begin{subfigure}[t]{0.32\textwidth}
  \begin{tikzpicture}[scale=0.6] \begin{axis}[
    xlabel = {$n$},
    ylabel = {Approximation Ratio},
    xmin=7, xmax=17,
    ymin=0, ymax=1,
		legend cell align=left,
		legend style={at={(0.68,0.4)}},
		y tick label style={
        /pgf/number format/.cd,
        fixed,
        fixed zerofill,
        precision=2,
        /tikz/.cd
    },
		error bars/y dir=both,
    error bars/y explicit]
		\pgfplotstableread{CSV/quadratic/m=0.5n/uniform2_quadratic_ours3__.csv}\ours
		\pgfplotstablecreatecol[copy column from table={CSV/quadratic/m=0.5n/uniform2_quadratic_ours3_std__.csv}{[index] 0}] {error} {\ours}
		\addplot [name path = our, blue, mark = *] table [x expr=8+2*\coordindex, y index=0, y error index=1] {\ours};
		\addlegendentry{Du~\cite{du2022lyapunov}  (our version)}
		\pgfplotstableread{CSV/quadratic/m=0.5n/uniform2_quadratic_theirs__.csv}\theirs
		\pgfplotstablecreatecol[copy column from table={CSV/quadratic/m=0.5n/uniform2_quadratic_theirs_std__.csv}{[index] 0}] {error} {\theirs}
		\addplot [name path = theirs, red, mark = triangle*, mark size=3pt] table [x expr=8+2*\coordindex, y index=0, y error index = 1] {\theirs};
		\addlegendentry{D{\"u}rr et al.~\cite{durr2021non}}
		\pgfplotstableread{CSV/quadratic/m=0.5n/uniform2_quadratic_theirs2__.csv}\theirstwo
		\pgfplotstablecreatecol[copy column from table={CSV/quadratic/m=0.5n/uniform2_quadratic_theirs2_std__.csv}{[index] 0}] {error} {\theirstwo}
		\addplot [name path = theirs, green!80!black, mark = square*] table [x expr=8+2*\coordindex, y index=0, y error index = 1] {\theirstwo};
		\addlegendentry{Du et al.~\cite{du2022improved}}
	\end{axis}\end{tikzpicture}
  \caption{$m=\lfloor0.5n\rfloor$}\label{fig:quad2}
\end{subfigure}\hfill
\begin{subfigure}[t]{0.32\textwidth}
  \begin{tikzpicture}[scale=0.6] \begin{axis}[
    xlabel = {$n$},
    ylabel = {Approximation Ratio},
    xmin=7, xmax=17,
    ymin=0, ymax=1,
		legend cell align=left,
		legend style={at={(0.68,0.4)}},
		y tick label style={
        /pgf/number format/.cd,
        fixed,
        fixed zerofill,
        precision=2,
        /tikz/.cd
    },
		error bars/y dir=both,
    error bars/y explicit]
		\pgfplotstableread{CSV/quadratic/m=n/uniform2_quadratic_ours3__.csv}\ours
		\pgfplotstablecreatecol[copy column from table={CSV/quadratic/m=n/uniform2_quadratic_ours3_std__.csv}{[index] 0}] {error} {\ours}
		\addplot [name path = our, blue, mark = *] table [x expr=8+2*\coordindex, y index=0, y error index=1] {\ours};
		\addlegendentry{Du~\cite{du2022lyapunov}  (our version)}
		\pgfplotstableread{CSV/quadratic/m=n/uniform2_quadratic_theirs__.csv}\theirs
		\pgfplotstablecreatecol[copy column from table={CSV/quadratic/m=n/uniform2_quadratic_theirs_std__.csv}{[index] 0}] {error} {\theirs}
		\addplot [name path = theirs, red, mark = triangle*, mark size=3pt] table [x expr=8+2*\coordindex, y index=0, y error index = 1] {\theirs};
		\addlegendentry{D{\"u}rr et al.~\cite{durr2021non}}
		\pgfplotstableread{CSV/quadratic/m=n/uniform2_quadratic_theirs2__.csv}\theirstwo
		\pgfplotstablecreatecol[copy column from table={CSV/quadratic/m=n/uniform2_quadratic_theirs2_std__.csv}{[index] 0}] {error} {\theirstwo}
		\addplot [name path = theirs, green!80!black, mark = square*] table [x expr=8+2*\coordindex, y index=0, y error index = 1] {\theirstwo};
		\addlegendentry{Du et al.~\cite{du2022improved}}
	\end{axis}\end{tikzpicture}
  \caption{$m=n$}\label{fig:quad1}
\end{subfigure}\hfill
\begin{subfigure}[t]{0.32\textwidth}
  \begin{tikzpicture}[scale=0.6] \begin{axis}[
    xlabel = {$n$},
    ylabel = {Approximation Ratio},
    xmin=7, xmax=17,
    ymin=0, ymax=1,
		legend cell align=left,
		legend style={at={(0.68,0.4)}},
		y tick label style={
        /pgf/number format/.cd,
        fixed,
        fixed zerofill,
        precision=2,
        /tikz/.cd
    },
		error bars/y dir=both,
    error bars/y explicit]
		\pgfplotstableread{CSV/quadratic/m=1.5n/uniform2_quadratic_ours3__.csv}\ours
		\pgfplotstablecreatecol[copy column from table={CSV/quadratic/m=1.5n/uniform2_quadratic_ours3_std__.csv}{[index] 0}] {error} {\ours}
		\addplot [name path = our, blue, mark = *] table [x expr=8+2*\coordindex, y index=0, y error index=1] {\ours};
		\addlegendentry{Du~\cite{du2022lyapunov}  (our version)}
		\pgfplotstableread{CSV/quadratic/m=1.5n/uniform2_quadratic_theirs__.csv}\theirs
		\pgfplotstablecreatecol[copy column from table={CSV/quadratic/m=1.5n/uniform2_quadratic_theirs_std__.csv}{[index] 0}] {error} {\theirs}
		\addplot [name path = theirs, red, mark = triangle*, mark size=3pt] table [x expr=8+2*\coordindex, y index=0, y error index = 1] {\theirs};
		\addlegendentry{D{\"u}rr et al.~\cite{durr2021non}}
		\pgfplotstableread{CSV/quadratic/m=1.5n/uniform2_quadratic_theirs2__.csv}\theirstwo
		\pgfplotstablecreatecol[copy column from table={CSV/quadratic/m=1.5n/uniform2_quadratic_theirs2_std__.csv}{[index] 0}] {error} {\theirstwo}
		\addplot [name path = theirs, green!80!black, mark = square*] table [x expr=8+2*\coordindex, y index=0, y error index = 1] {\theirstwo};
		\addlegendentry{Du et al.~\cite{du2022improved}}
	\end{axis}\end{tikzpicture}
  \caption{$m=\lfloor1.5n\rfloor$}\label{fig:quad3}
\end{subfigure}
\caption{Quadratic Programming with Uniform Distribution} \label{fig:quad11}
\end{figure}

\subsubsection{Exponential distribution} The other way to choose the matrices $\mH$ and $\mA$ is using an exponential distribution. Recall that given $\lambda > 0$, the exponential distribution $\exp(\lambda)$ is given by a density function assigning a density of $\lambda e^{-\lambda y}$ for every $y \geq 0$ and density $0$ for negative $y$ values. Then, $\mH \in \mathbb{R}^{n\times n}$ is randomly generated symmetric matrix whose entries are drawn independently from $-\exp(1)$, and $\mA \in \bR^{m \times n}$ is a randomly generated matrix whose entries are drawn independently from $\exp(0.25) + 0.01$.

For this way of generating $\mH$ and $\mA$, we repeated that same set of experiments as for the previous way of generating these matrices. The results of these experiments (averaged over $100$ repetitions) are depicted in Figure~\ref{fig:quad12}. Again, we note that the two sub-exponential time algorithms of D{\"u}rr et al.~\cite{durr2021non} and Du et al.~\cite{du2022improved} exhibit similar performance, and (our version) of the newer algorithm of Du~\cite{du2022lyapunov} significantly outperforms them, especially as the dimension $n$ grows.

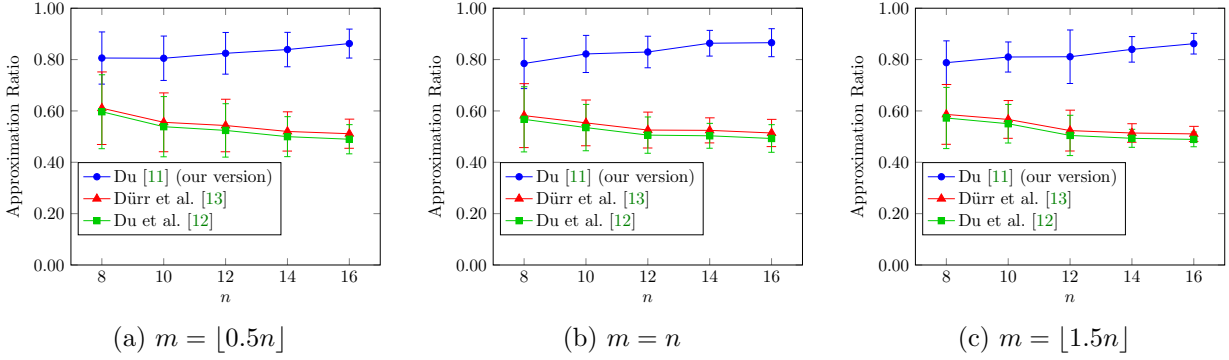
\begin{figure}[tb]
\begin{subfigure}[t]{0.32\textwidth}
  \begin{tikzpicture}[scale=0.6] \begin{axis}[
    xlabel = {$n$},
    ylabel = {Approximation Ratio},
    xmin=7, xmax=17,
    ymin=0, ymax=1,
		legend cell align=left,
		legend style={at={(0.68,0.4)}},
		y tick label style={
        /pgf/number format/.cd,
        fixed,
        fixed zerofill,
        precision=2,
        /tikz/.cd
    },
		error bars/y dir=both,
    error bars/y explicit]
		\pgfplotstableread{CSV/quadratic/m=0.5n/exponential2_quadratic_ours3__.csv}\ours
		\pgfplotstablecreatecol[copy column from table={CSV/quadratic/m=0.5n/exponential2_quadratic_ours3_std__.csv}{[index] 0}] {error} {\ours}
		\addplot [name path = our, blue, mark = *] table [x expr=8+2*\coordindex, y index=0, y error index=1] {\ours};
		\addlegendentry{Du~\cite{du2022lyapunov}  (our version)}
		\pgfplotstableread{CSV/quadratic/m=0.5n/exponential2_quadratic_theirs__.csv}\theirs
		\pgfplotstablecreatecol[copy column from table={CSV/quadratic/m=0.5n/exponential2_quadratic_theirs_std__.csv}{[index] 0}] {error} {\theirs}
		\addplot [name path = theirs, red, mark = triangle*, mark size=3pt] table [x expr=8+2*\coordindex, y index=0, y error index = 1] {\theirs};
		\addlegendentry{D{\"u}rr et al.~\cite{durr2021non}}
		\pgfplotstableread{CSV/quadratic/m=0.5n/exponential2_quadratic_theirs2__.csv}\theirstwo
		\pgfplotstablecreatecol[copy column from table={CSV/quadratic/m=0.5n/exponential2_quadratic_theirs2_std__.csv}{[index] 0}] {error} {\theirstwo}
		\addplot [name path = theirs, green!80!black, mark = square*] table [x expr=8+2*\coordindex, y index=0, y error index = 1] {\theirstwo};
		\addlegendentry{Du et al.~\cite{du2022improved}}
	\end{axis}\end{tikzpicture}
  \caption{$m=\lfloor0.5n\rfloor$}
\end{subfigure}\hfill
\begin{subfigure}[t]{0.32\textwidth}
  \begin{tikzpicture}[scale=0.6] \begin{axis}[
    xlabel = {$n$},
    ylabel = {Approximation Ratio},
    xmin=7, xmax=17,
    ymin=0, ymax=1,
		legend cell align=left,
		legend style={at={(0.68,0.4)}},
		y tick label style={
        /pgf/number format/.cd,
        fixed,
        fixed zerofill,
        precision=2,
        /tikz/.cd
    },
		error bars/y dir=both,
    error bars/y explicit]
		\pgfplotstableread{CSV/quadratic/m=n/exponential2_quadratic_ours3__.csv}\ours
		\pgfplotstablecreatecol[copy column from table={CSV/quadratic/m=n/exponential2_quadratic_ours3_std__.csv}{[index] 0}] {error} {\ours}
		\addplot [name path = our, blue, mark = *] table [x expr=8+2*\coordindex, y index=0, y error index=1] {\ours};
		\addlegendentry{Du~\cite{du2022lyapunov}  (our version)}
		\pgfplotstableread{CSV/quadratic/m=n/exponential2_quadratic_theirs__.csv}\theirs
		\pgfplotstablecreatecol[copy column from table={CSV/quadratic/m=n/exponential2_quadratic_theirs_std__.csv}{[index] 0}] {error} {\theirs}
		\addplot [name path = theirs, red, mark = triangle*, mark size=3pt] table [x expr=8+2*\coordindex, y index=0, y error index = 1] {\theirs};
		\addlegendentry{D{\"u}rr et al.~\cite{durr2021non}}
		\pgfplotstableread{CSV/quadratic/m=n/exponential2_quadratic_theirs2__.csv}\theirstwo
		\pgfplotstablecreatecol[copy column from table={CSV/quadratic/m=n/exponential2_quadratic_theirs2_std__.csv}{[index] 0}] {error} {\theirstwo}
		\addplot [name path = theirs, green!80!black, mark = square*] table [x expr=8+2*\coordindex, y index=0, y error index = 1] {\theirstwo};
		\addlegendentry{Du et al.~\cite{du2022improved}}
	\end{axis}\end{tikzpicture}
  \caption{$m=n$}
\end{subfigure}\hfill
\begin{subfigure}[t]{0.32\textwidth}
  \begin{tikzpicture}[scale=0.6] \begin{axis}[
    xlabel = {$n$},
    ylabel = {Approximation Ratio},
    xmin=7, xmax=17,
    ymin=0, ymax=1,
		legend cell align=left,
		legend style={at={(0.68,0.4)}},
		y tick label style={
        /pgf/number format/.cd,
        fixed,
        fixed zerofill,
        precision=2,
        /tikz/.cd
    },
		error bars/y dir=both,
    error bars/y explicit]
		\pgfplotstableread{CSV/quadratic/m=1.5n/exponential2_quadratic_ours3__.csv}\ours
		\pgfplotstablecreatecol[copy column from table={CSV/quadratic/m=1.5n/exponential2_quadratic_ours3_std__.csv}{[index] 0}] {error} {\ours}
		\addplot [name path = our, blue, mark = *] table [x expr=8+2*\coordindex, y index=0, y error index=1] {\ours};
		\addlegendentry{Du~\cite{du2022lyapunov}  (our version)}
		\pgfplotstableread{CSV/quadratic/m=1.5n/exponential2_quadratic_theirs__.csv}\theirs
		\pgfplotstablecreatecol[copy column from table={CSV/quadratic/m=1.5n/exponential2_quadratic_theirs_std__.csv}{[index] 0}] {error} {\theirs}
		\addplot [name path = theirs, red, mark = triangle*, mark size=3pt] table [x expr=8+2*\coordindex, y index=0, y error index = 1] {\theirs};
		\addlegendentry{D{\"u}rr et al.~\cite{durr2021non}}
		\pgfplotstableread{CSV/quadratic/m=1.5n/exponential2_quadratic_theirs2__.csv}\theirstwo
		\pgfplotstablecreatecol[copy column from table={CSV/quadratic/m=1.5n/exponential2_quadratic_theirs2_std__.csv}{[index] 0}] {error} {\theirstwo}
		\addplot [name path = theirs, green!80!black, mark = square*] table [x expr=8+2*\coordindex, y index=0, y error index = 1] {\theirstwo};
		\addlegendentry{Du et al.~\cite{du2022improved}}
	\end{axis}\end{tikzpicture}
  \caption{$m=\lfloor1.5n\rfloor$}
\end{subfigure}
\caption{Quadratic Programming with Exponential Distribution} \label{fig:quad12}
\end{figure}

\section{Conclusion}
In this work, we have considered the problem of maximizing a DR-submodular function over a general convex set in both the offline and the online (regret minimization) settings. For the online setting we provided the first polynomial time algorithm. Our algorithm matches the approximation guarantee of the only polynomial time algorithm known for the offline setting. Moreover, we presented a hardness result showing that this approximation guarantee is optimal for both settings. Finally, we have run experiments to study the empirical performance of both our algorithm and the (recently suggested) polynomial time offline algorithm. Our experiments show that both these algorithms outperform previous benchmarks. 

\appendix

\section{Proof of Lemma~\headerref{lem:norm_loss}} \label{sec:norm_less_proof}

In this section we prove Lemma~\ref{lem:norm_loss}, which we repeat here for convenience.
\lemNormLoss*
\begin{proof}
If $\norm{\vx}_\infty = 0$, then $\vx$ is the all zeros vector, and the lemma becomes trivial. Thus, we may assume in the rest of this proof that $\norm{\vx}_\infty > 0$. Let $\vz = \vx \vee \vy - \vy$. Then,
\begin{align} \label{eq:diff_z}
	F(\vx \vee \vy) - F(\vy)
	={} &
	\int_0^1 \left. \frac{dF(\vy + r \cdot \vz)}{dr}\right|_{r = t} dt
	=
	\int_0^1 \sum_{i = 1}^n \langle \vz, \nabla F(\vy + t \cdot \vz)\rangle dt\\\nonumber
	={} &
	\|\vx\|_\infty \cdot \int_0^{1/\|\vx\|_\infty} \sum_{i = 1}^n \langle \vz, \nabla F(\vy + \|\vx\|_\infty \cdot t' \cdot \vz)\rangle dt'\\\nonumber
	\geq{} &
	\|\vx\|_\infty \cdot \int_0^{1/\|\vx\|_\infty} \sum_{i = 1}^n \langle \vz, \nabla F(\vy + t' \cdot \vz)\rangle dt'
	\enspace,
\end{align}
where the last equality holds by changing the integration variable to $t' = t / \|\vx\|_\infty$, and the inequality follows from the DR-submodularity of $F$ because $\vy + t' \cdot \vz \in [0, 1]^n$. To see that the last inclusion holds, note that, for every $i \in [n]$, if $x_i \leq y_i$, then $y_i + t' \cdot z_i = y_i \leq 1$, and if $x_i \geq y_i$, then
\[
	y_i + t' \cdot z_i
	\leq
	y_i + \frac{z_i}{\|\vx\|_\infty}
	=
	y_i + \frac{x_i - y_i}{\|\vx\|_\infty}
	\leq
	\frac{x_i}{\|\vx\|_\infty}
	\leq
	1
	\enspace.
\]

Observe now that we also have
\begin{align*}
	\int_0^{1/\|\vx\|_\infty} \sum_{i = 1}^n \langle \vz, \nabla F(\vy + t' \cdot \vz)\rangle dt'
	={} &
	\int_0^{1/\|\vx\|_\infty} \left. \frac{dF(\vy + r \cdot \vz)}{dr}\right|_{r = t'} dt'\\
	={} &
	F\left(\vy + \frac{\vz}{\|\vx\|_\infty}\right) - F(\vy)
	\geq
	- F(\vy)
	\enspace,
\end{align*}
where the inequality follows from the non-negativity of $F$. The lemma now follows by plugging this inequality into Inequality~\eqref{eq:diff_z}, and rearranging.
\end{proof}
\toggletrue{proofsAppendix}
\section{Missing Proofs of Section~\headerref{sec:Hardness}} \label{app:missing_inapproximability}

\subsection{Proof of Claim~\headerref{clm:partial_derivatives_bound}} 

In this section we prove Claim~\ref{clm:partial_derivatives_bound}, which we repeat here for convenience.
\clmPartialDerivativesBound*
\begin{proof}
Recall that $\hat{F}_k$ and $\hat{G}_k$ are the functions $\hat{F}$ and $\hat{G}$ whose existence is guaranteed by Lemma~\ref{lem:original_continuous_versions} for $f = f_k$. The functions $\hat{F}$ and $\hat{G}$ are obtained in the proof of Lemma~\ref{lem:original_continuous_versions} in a series of steps involving multiple intermediate functions. The first of these functions are $F$ (the multilinear extension of $f$), the function $G(\vx) = F(\bar{\vx})$ and the function $H(\vx) = F(\vx) - G(\vx)$. The proof of Lemma~3.5 of~\cite{vondrak2013symmetry} shows that the absolute values of the second partial derivatives of these functions are bounded by $4M$, $4M$ and $8M$, respectively, where $M$ is the maximum value that the function $f$ can take. Since in our case $f$ is $f_k$, the maximum value it can take is $k$, and therefore, the absolute values of the second partial derivatives of all three functions can be upper bounded by $8k$.

The next function we consider is a function denoted by $\tilde{F}$ in the proof of Lemma~\ref{lem:original_continuous_versions}. The proof of Lemma~3.8 of~\cite{vondrak2013symmetry} shows that for every two elements $u, v \in \cN$, this function obeys almost everywhere the inequality
\[
	\left|\frac{\partial^2 \tilde{F}(\vx)}{\partial u \partial v} - \frac{\partial^2 F(\vx)}{\partial u \partial v} + \phi(D(\vx)) \cdot \frac{\partial^2 H(\vx)}{\partial u \partial v} \right|
	\leq
	512M |\cN|  \alpha
	=
	\frac{512\eps'}{2000 |\cN|^2}
	\leq
	1
	\enspace,
\]
where $\phi$ is a function defined by~\cite{vondrak2013symmetry} whose range is $[0, 1]$, $D(\vx)$ is another function defined by~\cite{vondrak2013symmetry} and $\alpha = \eps' / (2000 M |\cN|^3)$. Since $|\phi(D(\vx))| \leq 1$, the last inequality implies that the absolute values of the second partial derivatives of $\tilde{F}$ are upper bounded by $16k + 1$ because the second partial derivatives of $F$ and $H$ have absolute values bounded by $8k$.

The functions $\hat{F}$ and $\hat{G}$ are obtained from $\tilde{F}$ and $G$, respectively, by adding $256M|\cN|\alpha J(\vx) = \frac{256\eps'}{2000|\cN|^2} \cdot J(\vx)$, where
\[
	J(\vx) = |\cN|^2 + 3|\cN| \|\vx\|_1 - (\|\vx\|_1)^2
	\enspace.
\]
Since the second order partial derivatives of $J(\vx)$ are all $-2$, and the coefficient of $J(\vx)$ is $\frac{256\eps'}{2000|\cN|^2} \leq 1/2$, adding $\frac{256\eps'}{2000|\cN|^2} \cdot J(\vx)$ cannot increase the absolute value of the second order partial derivatives by more than $1$.
\end{proof}

\subsection{Proof of Lemma~\headerref{lem:scrambled_objectives}} 

In this section we prove Lemma~\ref{lem:scrambled_objectives}, which we repeat here for convenience.
\lemScrambledObjectives*
\begin{proof}
We prove the lemma below for $\bar{F}_{k, \vsigma}$. The proof for $\bar{G}_{k, \vsigma}$ is analogous. The non-negativity of $\bar{F}_{k, \vsigma}$ follows immediately from their definitions and the non-negativity of $\hat{F}_k$ and $\hat{G}_k$. Furthermore, by the chain-rule, for every pair of $i \in [k]$ and $j \in [\ell]$, we have
\begin{equation} \label{eq:first_order_formula}
	\frac{\partial \bar{F}_{k, \vsigma}(\vx)}{\partial \vx_{a_{i, j}}}
	=
	\frac{1}{\ell} \cdot \left.\frac{\partial \hat{F}_k(\vz)}{\partial \vz_{a_{\sigma_j(i)}}}\right|_{\vz = \vx^{(\vsigma)}}
	\qquad
	\text{and}
	\qquad
	\frac{\partial \bar{F}_{k, \vsigma}(\vx)}{\partial \vx_{b_{i, j}}}
	=
	\frac{1}{\ell} \cdot \left.\frac{\partial \hat{F}_k(\vz)}{\partial \vz_{b_{\sigma_j(i)}}}\right|_{\vz = \vx^{(\vsigma)}}
	\enspace.
\end{equation}
Thus, the continuous differentiability of $\hat{F}_k$ implies that $\bar{F}_{k, \vsigma}$ is also continuously differentiable.

Taking the derivative of the last equalities with respect to $a_{i', b'}$ for another pair $i' \in [k], j' \in [\ell]$, the chain-rule gives us the equalities
\[
	\frac{\partial^2 \bar{F}_{k, \vsigma}(\vx)}{\partial \vx_{a_{i', j'}} \partial \vx_{a_{i, j}}}
	=
	\frac{1}{\ell^2} \cdot \left. \frac{\partial^2 \hat{F}_k(\vz)}{\partial \vz_{a_{\sigma_{j'}(i')}} \partial \vz_{a_{\sigma_j(i)}}} \right|_{\vz = \vx^{(\vsigma)}}
\]
and
\[
	\frac{\partial^2 \bar{F}_{k, \vsigma}(\vx)}{\partial \vx_{a_{i', j'}} \partial \vx_{b_{i, j}}}
	=
	\frac{1}{\ell^2} \cdot \left. \frac{\partial^2 \hat{F}_k(\vz)}{\partial \vz_{a_{\sigma_{j'}(i')}} \partial \vz_{b_{\sigma_j(i)}}} \right|_{\vz = \vx^{(\vsigma)}}
	\enspace.
\]
Since similar equalities hold also when we take the derivative of the equalities in Equation~\eqref{eq:first_order_formula} with respect to $b_{i', j'}$, the DR-submodularity of $\hat{F}_k$ implies the same property for $\bar{F}_{k, \vsigma}$.

It remains to bound the smoothness of $\bar{F}_{k, \vsigma}$. For every two vectors $\vx, \vy \in [0, 1]^{\cM_k}$, we have by Equation~\eqref{eq:first_order_formula} that
{\allowdisplaybreaks\begin{align*}
	&\|\nabla \bar{F}_{k, \vsigma}(\vx) - \nabla \bar{F}_{k, \vsigma}(\vy)\|_2^2
	=
	\sum_{i = 1}^k \sum_{j = 1}^\ell \left(\frac{1}{\ell} \cdot \left.\frac{\partial \hat{F}_k(\vz)}{\partial \vz_{a_{\vsigma_j(i)}}}\right|_{\vz = \vx^{(\vsigma)}} \mspace{-9mu}- \frac{1}{\ell} \cdot \left.\frac{\partial \hat{F}_k(\vz)}{\partial \vz_{a_{\sigma_j(i)}}}\right|_{\vz = \vy^{(\vsigma)}}\right)^{\mspace{-9mu}2} \\&+ \sum_{i = 1}^k \sum_{j = 1}^\ell \left(\frac{1}{\ell} \cdot \left.\frac{\partial \hat{F}_k(\vz)}{\partial \vz_{b_{\sigma_j(i)}}}\right|_{\vz = \vx^{(\vsigma)}} \mspace{-9mu}- \frac{1}{\ell} \cdot \left.\frac{\partial \hat{F}_k(\vz)}{\partial \vz_{b_{\sigma_j(i)}}}\right|_{\vz = \vy^{(\vsigma)}}\right)^{\mspace{-9mu}2}
	=
	\frac{1}{\ell} \cdot \sum_{i = 1}^k \left(\left.\frac{\partial \hat{F}_k(\vz)}{\partial \vz_{a_i}}\right|_{\vz = \vx^{(\vsigma)}} \mspace{-9mu}- \left.\frac{\partial \hat{F}_k(\vz)}{\partial \vz_{a_i}}\right|_{\vz = \vy^{(\vsigma)}}\right)^{\mspace{-9mu}2} \\&+\frac{1}{\ell} \cdot \sum_{i = 1}^k \left(\left.\frac{\partial \hat{F}_k(\vz)}{\partial \vz_{b_i}}\right|_{\vz = \vx^{(\vsigma)}} \mspace{-9mu}- \left.\frac{\partial \hat{F}_k(\vz)}{\partial \vz_{b_i}}\right|_{\vz = \vy^{(\vsigma)}}\right)^{\mspace{-9mu}2}
	=
	\frac{\|\nabla \hat{F}_k(\vx^{(\vsigma)}) - \nabla \hat{F}_k(\vy^{(\vsigma)})\|_2^2}{\ell}
	\leq
	\frac{\beta^2 \|\vx^{(\vsigma)} - \vy^{(\vsigma)}\|_2^2}{\ell}\\
	={} &
	\frac{\beta^2 \cdot \sum_{i = 1}^k [(\sum_{j = 1}^\ell \vx_{a_{\sigma_j(i), j}} - \sum_{j = 1}^\ell \vy_{a_{\sigma_j(i), j}})^2 + (\sum_{j = 1}^\ell \vx_{b_{\sigma_j(i), j}} - \sum_{j = 1}^\ell \vy_{b_{\sigma_j(i), j}})^2]}{\ell^3}
	\enspace,
\end{align*}}%
where $\beta$ is the smoothness parameter of $\hat{F}_k$, and the second equality holds since the entries of $\vsigma$ are permutations. Using Sedrakyan's inequality (or Cauchy–Schwarz inequality), we also have, for every $i \in [k]$,
\[
	\left(\sum_{j = 1}^\ell \vx_{a_{\sigma_j(i), j}} - \sum_{j = 1}^\ell \vy_{a_{\sigma_j(i), j}}\right)^{\mspace{-9mu}2}
	\leq
	\ell \cdot \sum_{j = 1}^\ell (\vx_{a_{\sigma_j(i), j}} - \sum_{j = 1}^\ell \vy_{a_{\sigma_j(i), j}})^2
\]
and
\[
	\left(\sum_{j = 1}^\ell \vx_{b_{\sigma_j(i), j}} - \sum_{j = 1}^\ell \vy_{b_{\sigma_j(i), j}}\right)^{\mspace{-9mu}2}
	\leq
	\ell \cdot \sum_{j = 1}^\ell (\vx_{b_{\sigma_j(i), j}} - \sum_{j = 1}^\ell \vy_{b_{\sigma_j(i), j}})^2
	\enspace.
\]

Combining all the above inequalities yields
\begin{align*}
	\|\nabla \bar{F}_{k, \vsigma}(\vx) - \nabla \bar{F}_{k, \vsigma}(\vy)\|_2
	\leq{} &
	\frac{\beta \cdot \sqrt{\sum_{i = 1}^k [\sum_{j = 1}^\ell (\vx_{a_{\sigma_j(i), j}} - \vy_{a_{\sigma_j(i), j}})^2 + \sum_{j = 1}^\ell (\vx_{b_{\sigma_j(i), j}} - \vy_{b_{\sigma_j(i), j}})^2}]}{\ell}\\
	={} &
	\frac{\beta \cdot \|\vx - \vy\|_2}{\ell}
	\enspace,
\end{align*}
which completes the proof of the lemma since the smoothness parameter $\beta$ of $\hat{F}_k$ is polynomial in $k$.
\end{proof}

\subsection{Proof of Claim~\headerref{clm:x_sigma_P}} 

In this section we prove Claim~\ref{clm:x_sigma_P}, which we repeat here for convenience.
\clmXSigmaP*
\begin{proof}
By the definition of $\cK_{h, k, \ell}$, the membership of $\vx$ in $\cK_{h, k, \ell}$ implies that for every $j \in [\ell]$ we must have $\vx^{(j)} \in \cP_{h, k}$. Thus, $\vx^{(j)}$ can be represented by a convex combination of the vectors $\vu, \vv^{(1)}, \vv^{(2)}, \dotsc, \vv^{(k)}$ as follows.
\[
	\vx^{(j)}
	=
	\sum_{i = 1}^k c_{i,j} \cdot \vv^{(j)} + d_j \cdot \vu
	\enspace.
\]
Similarly to the proof of Lemma~\ref{lem:max_values_preserved}, let us define $\sigma^{-1}_j(\vx^{(j)})$ to be the following vector. For every $i \in [k]$,
\[
	(\sigma_j^{-1}(\vx^{(j)}))_{a_i}
	=
	\vx^{(j)}_{a_{\sigma(i)}}
	\qquad
	\text{and}
	\qquad
	(\sigma_j^{-1}(\vx^{(j)}))_{b_i}
	=
	\vx^{(j)}_{b_{\sigma(i)}}
	\enspace.
\]

Using the above notation, we get
\begin{align*}
	\vx^{(\vsigma)}
	={} &
	\tfrac{1}{\ell} \sum_{j = 1}^\ell \sigma^{-1}_j(\vx^{(j)})
	=
	\tfrac{1}{\ell} \sum_{j = 1}^\ell \sigma^{-1}_j\left(\sum_{i = 1}^k c_{i,j} \cdot \vv^{(i)} + d_j \cdot \vu\right)\\
	={} &
	\tfrac{1}{\ell} \sum_{j = 1}^\ell \left[\sum_{i = 1}^k c_{i,j} \cdot \sigma^{-1}_j(\vv^{(i)}) + d_j \cdot \sigma^{-1}_j(\vu)\right]
	=
	\tfrac{1}{\ell} \sum_{j = 1}^\ell \left[\sum_{i = 1}^k c_{i,j} \cdot \vv^{(\sigma^{-1}_j(i))} + d_j \cdot \vu\right]\\
	={} &
	\sum_{i = 1}^k \frac{\sum_{j = 1}^\ell c_{\sigma_j(i), j}}{\ell} \cdot \vv^{(i)} + \frac{\sum_{j = 1}^\ell d_j}{\ell} \cdot \vu
	\enspace.
\end{align*}
The last step in the proof of the claim is to show that the rightmost side is a convex combination, which implies $\vx^{(\vsigma)} \in \cP_{h, k}$ by the definition of $\cP_{h, k}$. To see that this is indeed the case, we observe that the coefficients of all the vectors in this rightmost side are averages of non-negative numbers, and therefore, are non-negative as well. Furthermore,
\[
	\sum_{i = 1}^k \frac{\sum_{j = 1}^\ell c_{\sigma_j(i), j}}{\ell} + \frac{\sum_{j = 1}^\ell d_j}{\ell}
	=
	\tfrac{1}{\ell} \sum_{j = 1}^\ell \left[\sum_{i = 1}^k c_{\sigma_j(i), j} + d_j\right]
	=
	\tfrac{1}{\ell} \sum_{j = 1}^\ell \left[\sum_{i = 1}^k c_{i, j} + d_j\right]
	=
	\tfrac{1}{\ell} \sum_{j = 1}^\ell 1
	=
	1
	\enspace,
\]
where the second equality holds since $\sigma_j$ is a permutation for every $j \in \ell$.
\end{proof}

\bibliographystyle{plain}
\bibliography{DR-Max}

\end{document}